\documentclass[10pt]{scrartcl}
\usepackage[dvips]{color}
\usepackage{epsfig}

\usepackage{amsmath,amssymb,dsfont,paralist,graphicx}
\usepackage[thmmarks,amsmath,standard]{ntheorem}

\newtheorem{df}{Definition}[section]
\newtheorem{lm}[df]{Lemma}
\newtheorem{ex}[df]{Example}
\newtheorem{propos}[df]{Proposition}
\newtheorem{theo}[df]{Theorem}
\newtheorem{cor}[df]{Corollary}
\newtheorem{ob}[df]{Observation}

\newcommand{\nat}{\mathbb{N}}

\newcommand{\seml}{[\![}
\newcommand{\semr}{]\!]}

\newcommand{\pos}{\mathrm{pos}}

\newcommand{\maxrk}{{\mathrm{maxrk}}}

\newcommand{\Pred}{\mathrm{Pred}}

\newcommand{\BC}{\mathrm{BC}}
\newcommand{\REC}{\mathrm{REC}}
\newcommand{\STT}{\mathrm{STT}}
\newcommand{\lSTT}{\mathrm{l}\n\mathrm{STT}}
\newcommand{\tSTT}{\mathrm{t}\n\mathrm{STT}}
\newcommand{\tdSTT}{\mathrm{td}\n\mathrm{STT}}
\newcommand{\dSTT}{\mathrm{d}\n\mathrm{STT}}
\newcommand{\slSTT}{\mathrm{sl}\n\mathrm{STT}}
\newcommand{\nSTT}{\mathrm{n}\n\mathrm{STT}}
\newcommand{\lnSTT}{\mathrm{ln}\n\mathrm{STT}}

\newcommand{\VR}{\mathrm{VR}}

\newcommand{\A}{{\cal A}}
\newcommand{\B}{{\cal B}}

\newcommand{\M}{{\cal M}}
\newcommand{\N}{{\cal N}}
\newcommand{\F}{{\cal F}}
\newcommand{\G}{{\cal G}}

\newcommand{\PS}{{\cal P}}

\DeclareMathOperator{\rk}{rk}

\DeclareMathOperator{\dom}{dom}
\DeclareMathOperator{\grd}{grd}
\DeclareMathOperator{\lhs}{lhs}
\DeclareMathOperator{\rhs}{rhs}
\DeclareMathOperator{\range}{range}

\DeclareMathOperator{\n}{-}

\def\-{$-$}
\def\|{\hspace{1mm} | \hspace{1mm}}

\def\seq#1#2#3{#1_{#2},\ldots,#1_{#3}} 

\def\g0{\geq0}             

\def\ui#1{^{(#1)}}

\begin{document}

\title{Forward and Backward Application of Symbolic Tree Transducers}

\author{Zolt\'an F\"ul\"op$^a$\thanks{Research of this author was supported by the program T\'AMOP-4.2.1/B-09/1/KONV-2010-0005  of the Hungarian National
    Development Agency.} and Heiko Vogler$^b$\\
  {\small $^a$ Department of Foundations of Computer Science,
    University of Szeged} \\[-.5ex]
  {\small \'Arp\'ad t\'er 2., H-6720 Szeged, Hungary.
    {fulop@inf.u-szeged.hu}} \smallskip \\
  {\small $^b$ Faculty of Computer Science, Technische Universit\"at
    Dresden} \\[-.5ex]
  {\small Mommsenstr.~13, D-01062 Dresden, Germany.
    {Heiko.Vogler@tu-dresden.de}}}

\date{\today}

\maketitle

\sloppy

\begin{quote}{\bf Abstract:} We consider symbolic tree automata (sta) and symbolic tree transducers (stt). We characterize s-recognizable tree languages (which are the tree languages recognizable by sta) in terms of (classical) recognizable tree languages and relabelings. We prove that sta and the recently introduced variable tree automata are incomparable with respect to their recognition power.  We define symbolic regular tree grammars and characterize s-regular tree languages in terms of regular tree languages and relabelings. As a consequence, we obtain that s-recognizable tree languages are the same as s-regular tree languages.

We show that  the syntactic composition of two stt computes the composition of
the tree transformations computed by each stt, provided that (1) the first one is deterministic or the second one is linear and (2) the first one is total or the second is nondeleting.  We consider forward application and backward application of stt and prove  that the backward application of an stt to any s-recognizable tree language yields an s-recognizable tree language. We give a linear stt of which the range is not an s-recognizable tree language. We show that the  forward application of simple and linear stt preserves s-recognizability. As a corollary, we obtain  that the type checking problem of simple and linear stt and the inverse type checking problem of arbitrary stt is decidable.
\end{quote}



\section{Introduction}

\label{introduction} Symbolic tree automata (sta) and symbolic tree transducers (stt) were introduced in \cite{veabjo11a} and \cite{veabjo11b}. They differ from classical finite-state tree automata and tree transducers \cite{gecste84,gecste97} in that they work with trees over an infinite, unranked set of symbols. 
According to \cite{grukupshe10}, examples of systems with finite control and infinite source of data are software with integer parameters \cite{bouhabmay03}, datalog systems with infinite data domain \cite{bouhabjursig07}, and XML documents of which the leaves are associated with data values from some infinite domain
\cite{bracercomframan03}. It was mentioned in \cite{veabjo11a} that lifting the finite alphabet restriction is useful to enable efficient symbolic analysis. Symbolic transducers are useful for exploring symbolic solvers when performing basic automata-theoretic transformations \cite{veahoolivmolbjo12}. 

In this paper we provide new formal definitions of sta and stt which slightly differs from those given in \cite{veabjo11a,veabjo11b}. At the end of Sections \ref{sect:sta-def} and \ref{sect:stt-def} we will compare our definitions with the original ones.

Roughly speaking, an sta is a finite-state tree automaton \cite{don70}  except that the input trees are built up over an infinite set of labels. In order to ensure a finite description of the potentially infinite set of transitions we bind the maximal number of the successors of a any node
occurring in an input tree by an integer $k\in \nat$, and we employ finitely many unary Boolean-valued predicates over the set of labels.  Then every transition of a symbolic k-bounded tree automaton (s$k$-ta) has the form 
\[
(q_1 \ldots q_l , \varphi, q)
\]
where $0\le l\le k$, $q,q_1,\ldots,q_l$ are states, and $\varphi$ is a unary Boolean-valued predicate. Such a transition is applicable to a node if $\varphi$ holds for the label of that node.
The tree language $L(\A)$ recognized by an sta $\A$ is defined as the union of all tree languages $L(\A,q)$, where $q$ is a final state, and the family $(L(\A,q)\mid q \in Q)$ is defined inductively in the same way as for finite-state tree automata. A tree language is s$k$-recognizable if there is an s$k$-ta which recognizes this language, and it is s-recognizable if it is s$k$-recognizable for some $k\in \nat$.
An example of an s2-recognizable tree language is 
the set of all binary trees with labels taken from $\mathbb{N}$ such that every label is divisible by 2 or every label is divisible by 3 as, e.g., $2(4,6)$ or $3(15,18)$ (cf. Example \ref{ex:sta}).  

By restricting the set of labels to a ranked alphabet $\Sigma$ and just allowing, for every $\sigma \in \Sigma$, the characteristic mapping on $\{\sigma\}$ as predicate, we reobtain the classical finite-state tree automata. In \cite{veabjo11a} it was proved that bottom-up sta are determinizable, that the class of s-recognizable tree languages is closed under the Boolean operations, and that the emptiness problem for s-recognizable tree languages is decidable provided the emptiness problem in the Boolean algebra of predicates is decidable.

Similarly, an stt is a top-down tree transducer \cite{tha70,rou70,eng75} except that its input and output trees  are built up over potentially infinite sets of (resp., input and output) labels. 
In the same way as for  sta, we ensure finiteness by an a priori bound $k$ on the maximal number of the successors of a node and by using a finite set of unary predicates. The right-hand side of each rule of a symbolic $k$-bounded tree transducer (s$k$-tt) 
contains unary functions, rather than explicit output symbols as in top-down tree transducers. These functions are then applied to the current input label and thereby produce the output labels. More formally, a rule has the form
\[
q(\varphi(x_1,\ldots,x_l)) \rightarrow u
\]
where $0\le l\le k$, $q$ is a state, $\varphi$ is a unary Boolean-valued predicate over the set of input labels, $x_1,\ldots,x_l$ are the usual variables that represent input subtrees, and $u$ is a tree in which each internal node has at most $k$ successors and is labeled by a unary function symbol; the leaves of $u$ can be labeled alternatively by objects $q'(x_i)$ with state $q'$ and $x_i \in \{x_1,\ldots,x_l\}$. Clearly, the leaf labels of the form $q'(x_i)$ organize the recursive descent on the input tree as usual in a top-down tree transducer. The tree transformation  computed by an stt is defined in the obvious way by means of a binary derivation relation. 
For instance, there is a (nondeterministic) s2-tt which transforms each binary tree over $\mathbb{N}$ into a set of binary trees over $\mathbb{N}$ such that a subtree $n(\xi_1,\xi_2)$ of the input tree is transformed into $m(\xi_1',\xi_2')$ where
\begin{itemize}
\item $m = n$, and $\xi_1'$ and $\xi_2'$ are transformations of $\xi_1$ and $\xi_2$, respectively, or
\item $m = \frac{n}{6}$ if $n$ is divisible by 6, and both $\xi_1'$ and $\xi_2'$ are transformations of $\xi_1$
(cf. Example \ref{stt-example}). 
\end{itemize} 

 By restricting the predicates on the input labels to some ranked alphabet (as for sta above) and by only allowing unary functions such that each one produces a constant symbols from some ranked (output) alphabet, we reobtain top-down tree transducers.  

Since sta and stt can check and manipulate data from an infinite set, they can be considered as tools for analyzing and transforming trees as they occur, e.g., in XML documents. Thus, the theoretical investigation of sta and stt is motivated by practical problems as e.g. type checking and inverse type checking.

In this paper we further develop the theory of sta and stt. We prove  a characterization of s-recognizable tree languages in terms of (classical) recognizable tree languages and relabelings (Thm. \ref{th:char}). We compare the recognition power of sta with that of variable tree automata  \cite{menrah11} (also cf.  \cite{grukupshe10}). More specifically, we characterize the tree language recognized by a variable tree automaton by the union of infinitely many s-recognizable tree languages  (Prop. \ref{prop:vta-sta}) and we show that sta and variable tree automata are incomparable  with respect to recognition power (cf. Thm. \ref{var-vs-sym-theo}).
Moreover, as a generalization of (classical) regular tree grammars \cite{bra69} we introduce symbolic regular tree grammars and characterize s-regular tree languages in terms of
regular tree languages and relabelings (Thm. \ref{th:char-reg}). As a corollary, we obtain that s-recognizable tree languages are the same as s-regular tree languages (Thm. \ref{th:rec=reg}).

 For stt we recall the concept of the syntactic composition from \cite{veabjo11b}. We show that  syntactic composition of two stt $\cal M$ and $\cal N$ computes the composition of the tree transformations computed by $\cal M$ and $\cal N$, provided that (1) $\cal M$ is deterministic or $\cal N$ is linear or (2) $\cal M$ is total or $\cal N$ is nondeleting (Thm. \ref{comp-lemma}). Hereby, we generalize Baker's classical result  \cite[Thm. 1]{bak79}. 

Finally, we consider forward application and backward application of stt; these investigations are motivated by the (inverse) type checking problem (see among others \cite{milsucvia03,alomilnevsucvia03,engman03,manberpersei05}). 
We show that the backward application of an s$k$-tt (which is the application of its inverse) to any s$k$-recognizable tree language yields an s$k$-recognizable tree language (Thm. \ref{backward-theo}). It is well-known that the forward application of linear top-down tree transducers preserves recognizability of tree languages (see e.g. \cite{tha69} or \cite[Ch. IV, Cor. 6.6]{gecste84}). It is surprising that for stt the corresponding result does not hold, in fact there is a linear s$k$-tt of which the range is not an s$k$-recognizable tree language (Lm. \ref{range-lin-lemma}). However, the application of simple and linear stt preserve s-recognizability (Thm. \ref{slin-theo}).
As a corollary, we obtain that the type checking problem of simple and linear stt, as well as, the inverse type checking problem of arbitrary stt is decidable (Thm. \ref{thm:type-checking}).

Since the theory of sta and stt is based on concepts which are slightly different from the foundations of classical finite-state tree automata and tree transducers, we list them in detail is Section  \ref{prel-section}.

\section{Preliminaries}\label{prel-section}

\subsection{General}

The set of nonnegative integers is denoted by $\nat.$ 

For a set $A$, we denote by $|A|$ and ${\cal P}(A)$ the cardinality and
the set of all subsets of $A$. Moreover, we denote by $\iota_A$ the identical mapping over $A$. For a set $I$, an \emph{$I$-indexed family over $A$} is a
mapping $f:I \rightarrow A$. We denote the family $f$ also by $(f_i
\mid i \in I)$. 

Let $\rho \subseteq A \times B$ be a relation. For every $A'\subseteq A$, we define $\rho(A')= \{b\in B\mid  (a,b)\in \rho \text{ for some } a\in A'\}$.
For another relation $\sigma \subseteq B \times C$, the composition of $\rho$ and $\sigma$ is the relation
$\rho \circ \sigma =\{(a,c) \mid \exists(b\in B) : (a,b)\in \rho \text{ and } (b,c)\in \sigma \}$.
The reflexive and transitive closure of a relation $\rho \subseteq A \times A$ is denoted by $\rho^*$.

\subsection{Trees}\label{sect:trees}

In this paper we mainly consider trees over a nonempty and unranked set. We note that our concept of a tree differs from that of \cite{veabjo11a,veabjo11b} in that we do not consider the empty tree as the base of the inductive definition. 

Let $U$ be a (possibly infinite) nonempty set, called the set of {\em labels}, and $Y$ a further set. The {\em set of trees over $U$} (or: {\em $U$-trees}) {\em indexed by $Y$}, denoted by $T_U(Y)$, is the smallest subset $T$ of $(U\cup Y\cup\{(,)\}\cup\{,\})^*$ such that  (i) $(U\cup Y)\subseteq T$, and (ii) if $a \in U$ and $\xi_1,\ldots ,\xi_l \in T$ with $l \ge 1$, then $a(\xi_1, \ldots,\xi_l) \in T$.  If $Y=\emptyset$, then we write $T_U$ for $T_U(Y)$. A {\em tree language over $U$}  (or: {\em $U$-tree language}) is any subset of $T_U$.

Let $Q$ be a set with $Q\cap U=\emptyset$.  Then we denote by $Q(T_U(Y))$ the subset $\{q(\xi) \mid q\in Q, \xi \in T_U(Y)\}$ of $T_{Q\cup U}(Y)$.

We define the set of {\em positions in a  $U$-tree} by means
of the mapping $\mathrm{pos}: T_U(Y)\rightarrow {\cal P}(\nat^*)$ inductively on the argument $\xi \in T_U(Y)$ as follows: (i) if  $\xi \in (U\cup Y)$, then
$\pos(\xi)=\{\varepsilon\}$, and (ii) if $\xi = a(\xi_1,\ldots,\xi_l)$ for some $a \in U$, $l \ge 1$
and $\seq \xi1l\in T_U(Y)$, then $\mathrm{pos}(\xi) = \{\varepsilon\} \cup \{iv \, | \, 1 \le i \le l, v \in \mathrm{pos}(\xi_i)\}$. 

For every $\xi \in T_U(Y)$ and $w \in \mathrm{pos}(\xi)$, the {\it
  label of $\xi$ at   $w$}, denoted by $\xi(w) \in (U\cup Y)$, the {\it
  subtree of $\xi$ at   $w$}, denoted by $\xi|_w \in T_U(Y)$, 
and the {\it
  rank at $w$}, denoted by $rk_\xi(w)\in \nat$, are defined inductively as follows: (i) if  $\xi \in (U\cup Y)$, then
$\xi(\varepsilon) = \xi|_\varepsilon = \xi$, and $rk_\xi(\varepsilon) = 0$, and
(ii) if $\xi = a(\xi_1,\ldots,\xi_l)$ for some $a \in U$,
$l \ge 1$ and $\xi_1,\ldots,\xi_l \in T_U(Y)$,  then
$\xi(\varepsilon)$ = $a$, $\xi|_\varepsilon = \xi$, and $rk_\xi(\varepsilon) = l$, and if  $1
\le i \le l$ and $w = iv$, then $\xi(w) = \xi_i(v)$, $\xi|_w = \xi_i|_v$, and $rk_\xi(w) = rk_{\xi_i}(v)$.

Let $\xi \in T_U(Y)$ be a tree. For any $V \subseteq U$, we define $\pos_V(\xi)=\{ w\in \pos(\xi) \mid \xi(w) \in V\}$. If $V=\{a\}$, then we write just $\pos_a(\xi)$ for $\pos_V(\xi)$. 
Moreover, for every $\zeta \in T_U(Y)$ and $w\in \pos(\xi)$, we denote by $\xi[\zeta ]_w$ the tree which is obtained by replacing the subtree $\xi|_w$ by $\zeta$.

We will consider trees with variables and the substitution of trees for variables. For this, let $X = \{x_1,x_2,\ldots\}$ be an infinite set of variables, disjoint with $U$, and let  $X_l = \{x_1,\ldots,x_l\}$ for every $l \in \nat$. For trees $\xi \in T_U(X_l)$ and $\zeta_1,\ldots,\zeta_l \in  T_U(Y)$, we denote by
$\xi[\zeta_1,\ldots,\zeta_l]$ the tree which we obtain by replacing  every occurrence of $x_i$ by $\zeta_i$ for every $1\le i\le l$. We note that $\xi[\zeta_1,\ldots,\zeta_l] \in T_U(Y)$. Moreover, we denote by $C_U(X_l)$ the set of trees in $T_U(X_l)$ in which each variable $x_i$ occurs exactly once and the order of variables from left to right is $x_1,\ldots,x_l$. We call the elements of 
$C_U(X_l)$ {\em $l$-contexts}.

Finally, let $\xi \in T_U(Y)$ and $k \in \nat$. 
We define the \emph{rank} $\rk(\xi)$ of $\xi$ to be $\rk(\xi) = \max\{\rk_\xi(w) \mid w \in \pos(\xi)\}$ and  we say that {\em $\xi$ is $k$-bounded} if $\rk(\xi) \le k$. We denote 
the set of all $k$-bounded $U$-trees indexed by $Y$ by $T_U^{(k)}(Y)$. Clearly, $T_U^{(k)}(Y)\subset T_U^{(k+1)}(Y)$.
A \emph{$k$-bounded $U$-tree language} (or: $(U,k)$-tree language) is a subset of $T_U^{(k)}$. 
A $U$-tree language $L$ is \emph{bounded} if there is a $k \in \nat$ such that $L$ is $k$-bounded. 
Moreover, we define the set of {\em $k$-bounded $l$-contexts} to be $C_U\ui k(X_l)=C_U(X_l)\cap T_U^{(k)}(X_l)$.

\begin{quote}\it In this paper $U$, $V$, and $W$ will always denote arbitrary nonempty sets unless specified otherwise.
\end{quote}

\subsection{Tree transformations}\label{sect:tree-trans}

Let $k \in \nat$. A \emph{$k$-bounded tree transformation} (or: \emph{$k$-tree transformation}) is a mapping $\tau : T_U^{(k)} \to {\cal P}(T_{V}^{(k)})$ (or: alternatively, a relation $\tau \subseteq T_U^{(k)}
\times T_{V}^{(k)}$). A \emph{tree transformation} is a $k$-tree transformation for some $k \in \nat$.  If for every $\xi \in T_U^{(k)}$, there is exactly one $\zeta \in T_{V}^{(k)}$ such that $(\xi,\zeta) \in \tau$ (i.e., $\tau$~is a
mapping), then we also write $\tau \colon T_U^{(k)} \to T_{V}^{(k)}$. The {\em inverse $\tau^{-1}$}, the {\em domain $\dom(\tau)$},
and the {\em range $\range(\tau)$} of a tree transformation $\tau$ are defined in the standard way.

Let $\tau  \subseteq T_U^{(k)} \times T_{V}^{(k)}$ be a tree transformation, $L \subseteq T_U^{(k)}$ and $L' \subseteq T_{V}^{(k)}$ tree languages.
The \emph{forward   application} (or just: {\em application}) \emph{of~$\tau$ to $L$} is the tree language $\tau(L)= \{
\zeta \in T_{V}^{(k)}\mid \exists (\xi \in L) : (\xi,\zeta)\in
\tau\}$. The \emph{backward application} of~$\tau$ to $L'$ is the tree
language $\tau^{-1}(L')$ (which is
the forward application of $\tau^{-1}$ to $L'$).

We extend the above concepts and the composition of tree transformations to classes of tree transformations and classes of tree languages in a natural way. For instance, if $\cal C$ and $\cal C'$ are classes of $k$-tree transformations, and $\cal L$ is a class of $k$-tree languages, then we define ${\cal C}\circ {\cal C'} = \{\tau\circ  \sigma \mid \tau \in {\cal C} \text{ and }  \sigma \in {\cal C'}\}$ and ${\cal C}({\cal L})= \{ \tau(L) \mid \tau \in {\cal C} \text{ and }  L \in {\cal L}\}$.

A \emph{relabeling} is a mapping $\tau: U \rightarrow {\cal P}(V)$
such that $\tau(a)$ is recursive and it is decidable if $\tau(a)=\emptyset$ for every $a \in U$; it
is called \emph{deterministic} if $\tau(a)$ is a singleton for every
$a \in U$. Let $k \in \nat$. The \emph{$k$-tree relabeling (induced by
  $\tau$)} is the mapping $\tau': T_U^{(k)} \rightarrow {\cal
  P}(T_{V}^{(k)})$, defined by $$\tau'(a(\xi_1,\ldots,\xi_l))=\{b(\zeta_1,\ldots,\zeta_l)\mid b\in \tau(a) \text{ and } \zeta_i\in \tau'(\xi_i) \text{ for } 1\le i\le l \}.$$
Then the mapping $\tau'$ is extended to $\tau'': {\cal P}(T_U^{(k)}) \rightarrow {\cal P}(T_{V}^{(k)})$
by $\tau''(L) = \bigcup_{\xi\in L} \tau'(\xi)$ for every $L \in {\cal
  P}(T_U^{(k)})$. 

We note that the composition of two $k$-tree relabelings $\tau_1'$ and
$\tau_2'$  is again a $k$-tree relabeling. In fact, if $\tau_1: U \rightarrow {\cal P}(V)$
and $\tau_2:  V \rightarrow {\cal P}(W)$, then $\tau_1 \circ \tau_2$
induces $\tau_1' \circ \tau_2'$. 
In the sequel, we drop the primes from $\tau'$ and $\tau''$ and
identify both mappings with $\tau$. 

\subsection{Predicates and label structures}

 A (unary) \emph{predicate
  over $U$} is a mapping $\varphi: U \rightarrow \{0,1\}$. We denote
by $\Pred(U)$ the set of all predicates over $U$. 
Let $\varphi \in \Pred(U)$ be a predicate. We introduce the notation $\seml
\varphi \semr$ for $\{a \in U \mid \varphi(a) =1\}$.

We define the operations
$\neg$, $\wedge$, and $\vee$ over $\Pred(U)$ in the obvious way and extend $\wedge$ and $\vee$
to finite families $(\varphi_i \mid i\in I)$ of predicates in $\Pred(U)$.  In particular,
$\seml \bigwedge_{i\in \emptyset} \varphi_i\semr = U$ and $\seml \bigvee_{i\in \emptyset} \varphi_i\semr = \emptyset$.

Let $\Phi \subseteq \Pred(U)$ be a finite set of  recursive predicates such that
$\seml\varphi \semr=\emptyset$ is decidable for every $\varphi \in \Phi$.
We call the pair $(U,\Phi)$ a \emph{label structure}.
 The \emph{Boolean closure of~$\Phi$}, denoted
by $\BC(\Phi)$, is the smallest set $B \subseteq  \Pred(U)$ such that 
\begin{enumerate}
\item[(i)] $\Phi \subseteq B$,
\item[(ii)] $\bot, \top \in B$ where $\top(a) = 1$ and $\bot(a) = 0$ for
  every $a \in U$,  and
\item[(iii)] for every $\varphi, \psi \in B$, the predicates $\neg \varphi$,
  $\varphi \wedge \psi$, and $\varphi \vee \psi$ are in $B$.
\end{enumerate}
It is clear that
$(\BC(\Phi),\wedge,\vee,\neg,\bot,\top)$ is a Boolean algebra for every $\Phi \subseteq \Pred(U)$.

\subsection{Tree automata, tree grammars, tree transducers}

We assume that the reader is familiar with the basic concepts of the theory of (classical) tree automata and tree transducers which can be found among others in \cite{gecste84,gecste97} and \cite{comdaugiljaclugtistom97}. In particular, we freely use the concept of a ranked alphabet, a tree language over a ranked alphabet, a finite-state tree automaton, a recognizable tree language, a regular tree grammar, a regular tree language, a top-down tree transducer, and of a tree transformation. Here we recall only some notations.

A {\em ranked alphabet} is a finite set $\Sigma$ equipped with a rank mapping $\rk_\Sigma: \Sigma \to \nat$. We define $\Sigma_l =\{ \sigma \in \Sigma \mid \rk_\Sigma(\sigma)=l \}$ ($l\ge 0$) and $\maxrk(\Sigma)=\max\{\rk_\Sigma(\sigma) \mid \sigma \in \Sigma \}$. 
It is clear that every tree $\xi \in T_\Sigma$ is $\maxrk(\Sigma)$-bounded.

A \emph{finite-state tree automaton} is a system $\A = (Q,\Sigma,\delta,F)$, where $Q$ is a finite, nonempty  set (states), $\Sigma$  is a ranked alphabet,  $\delta = (\delta_\sigma \mid \sigma \in \Sigma)$ is the family of sets of transitions, i.e., $\delta_\sigma \subseteq Q^l \times Q$ for every $l \in \nat$ and $\sigma \in \Sigma$ with $\rk_\Sigma(\sigma) =l$, and $F \subseteq Q$ is the set of final states. The set of trees recognized by $\A$ is denoted by $L(\A)$.
A tree language $L \subseteq T_\Sigma$ is {\em recognizable} if there is a finite-state tree automaton ${\cal A}$ such that $L = L({\cal A})$.

A \emph{regular tree grammar} is a tuple $\G = (Q,\Sigma,q_0,R)$ where $Q$ is a finite set of states\footnote{Usually these symbols are called nonterminals; but since this notion leads to misunderstandings in the application area of natural language processing, we prefer to call these symbols states.}, $\Sigma$ is a ranked alphabet, $q_0 \in Q$ (initial state), and  $R$ is a finite set of rules of the form $q \rightarrow u$ with $q \in Q$ and $u \in T_\Sigma(Q)$. The derivation relation induced by $\G$ and the tree language generated by $\G$ are denoted by $\Rightarrow_\G$ and $L(\G)$, respectively.
We will also consider {\em reduced} regular tree grammars and regular tree grammars in {\em normal form} in the sense of \cite{comdaugiljaclugtistom97}.

\section{Symbolic tree automata}

In this section we formalize our adaptation of the concept of a symbolic tree automaton from \cite{veabjo11a} and compare our model with the original one. Then we prove basic properties of sta. Finally, we compare the recognition capacity of sta with that of variable tree automata.

\subsection{Definition of sta}\label{sect:sta-def}

\begin{df} \rm Let $k \in \nat$. A \emph{symbolic $k$-bounded  tree automaton} (s$k$-ta) is a tuple $\A = (Q,U,\Phi,F,R)$ where 
\begin{itemize}
\item $Q$ is a finite, nonempty set (states),
\item $(U,\Phi)$ is a label structure,
\item $F \subseteq Q$ (set of final states), and 
\item $R$ is a finite set of rules of the form $(q_1 \ldots q_l, \varphi, q)$
where $0 \le l \le k$, $q_1,\ldots,q_l,q \in Q$, and $\varphi \in \BC(\Phi)$. 
\end{itemize} 
\end{df}
Let $\rho = (q_1 \ldots q_l, \varphi, q) \in R$. We call $(q_1 \ldots q_l)$ the {\em left-hand side}, $\varphi$ the \emph{guard}, and $q$  the {\em right-hand side} of the rule $\rho$, and denote them by $\lhs(\rho)$, $\grd(\rho)$, and $\rhs(\rho)$ respectively. 
Clearly, every s$k$-ta is an s$(k+1)$-ta.  By a {\em symbolic tree automaton} (sta) we mean an s$k$-ta for some $k \in \nat$.

For every $q\in Q$, we define the tree language $L(\A,q)\subseteq
T_U\ui k$ recognized by $\A$ in state $q$, as follows. The family  $(L(\A,q) \mid q \in Q)$ is the smallest $Q$-family $(L_q \mid q \in Q)$ of tree languages such that 
\begin{enumerate}
\item[(i)] if $a \in U$, $(\varepsilon,\varphi,q) \in R$, and $a \in \seml \varphi \semr$, then $a \in L_q$, and 
\item[(ii)] if $a \in U$, $(q_1 \ldots q_l, \varphi, q) \in R$ with $1\le l
  \le k$ and $a \in \seml \varphi\semr$, and $\xi_1 \in L(\A,q_1)$, \ldots, $\xi_l \in L(\A,q_l)$, then $a(\xi_1,\ldots,\xi_l) \in L_q$.
\end{enumerate}
The condition that all predicates in $\Phi$ (and hence in $\BC(\Phi)$) are recursive ensure that we can decide whether $\xi \in L(\A,q)$ for every
$q\in Q$ and $\xi \in T_U\ui k$.

The \emph{tree language recognized by $\A$}, denoted by $L(\A)$, is the set 
\[
L(\A) = \bigcup_{q \in F} L(\A,q)\enspace.
\]
A tree language $L \subseteq T_U^{(k)}$ is {\em symbolically $k$-recognizable} (s$k$-recognizable) if there is an s$k$-ta $\A$ such that $L(\A) = L$. We denote the class of all s$k$-recognizable $U$-tree languages by $\REC^{(k)}(U)$. Moreover, we call a tree language {\em s-recognizable} if it is s$k$-recognizable for some $k\in \nat$.

Two s$k$-ta $\A$ and $\B$ are \emph{equivalent} if $L(\A) = L(\B)$.

\begin{ex}\rm\label{ex:sta} We give an example of an  sta. For this we consider the set $U = \nat$ and the $2$-bounded tree language 
\[
\begin{array}{ll}
L = & \big\{  \xi \in T_\nat^{(2)} \mid \text{$\xi$ is binary and }\\[2mm]
&\big((\forall w \in \pos(\xi): \xi(w) \text{ is divisible by $2$}) \vee (\forall w \in \pos(\xi): \xi(w) \text{ is divisible by $3$}\big)\big\}
\end{array}
\]
where a tree $\xi$ is binary if $\rk_\xi(w) \in \{0,2\}$ for every $w \in \pos(\xi)$. For instance, the trees $2(4,6)$ and $3(15,18)$ are in $L$.

The following s$2$-ta $\A = (Q,\nat,\Phi,F,R)$ recognizes $L$:
\begin{itemize} 
\item $Q = F = \{2,3\}$, 
\item $\Phi = \{\mathrm{div}(2),\mathrm{div}(3)\}$ with $\seml \mathrm{div}(i) \semr = \{n \in \nat \mid n \text{ is divisible by $i$}\}$, 
\item for every $i \in \{2,3\}$ the transitions $(\varepsilon, \mathrm{div}(i), i)$ and $(i\, i, \mathrm{div}(i), i)$ are in $R$.
\end{itemize}
For instance, $6(12,18) \in L(\A,2) \cap L(\A,3)$. 
\end{ex}

Our definition of sta slightly differs from the one in \cite{veabjo11a} in the following two points.

1. They fix a Boolean algebra $B$ of predicates in advance, and then they make a theory of sta only over $B$. We are free to choose predicates whenever we need them.

2. In \cite{veabjo11a} no bound on the number of successors of nodes is mentioned. In our definition we put an explicit bound on this number in order to guarantee closure of s$k$-recognizable tree languages under complement.

We note that this closure under complement is not discussed clearly in \cite{veabjo11a}. The root of the ambiguity is that the complement of a tree language, appearing in Prop. 3 of that paper, is not defined. If the complement of a tree language $L$ is meant to be ${\cal U}^{\mathrm{T}\langle\sigma\rangle} \setminus L$ (as maybe is suggested by the definition of the complement of a predicate \cite[p.146]{veabjo11a}), which corresponds to $T_U \setminus L$ in our notation, then 
the class of s-recognizable tree languages is not closed under complement as stated in \cite[Prop. 3, Thm. 1]{veabjo11a}. This can be seen easily as follows. Let $L$ be an  s-recognizable tree language (in the sense of \cite{veabjo11a} or of the present paper). Then obviously $L$ is bounded, while  the tree language ${\cal U}^{\mathrm{T}\langle\sigma\rangle} \setminus L$
is not bounded. Hence the latter cannot be s-recognizable. 

However, if we define the complement of a $k$-bounded tree language $L$ with respect to $T_U^{(k)}$, i.e., to be $T_U^{(k)} \setminus L$,
then the class of s$k$-recognizable tree languages is closed under complement (by using the appropriate adaptations of \cite[Prop. 3, Thm. 1]{veabjo11a}.

\subsection{Basic properties}

Here we give a characterization of s-recognizable tree languages in terms of (classical) recognizable tree languages and tree relabelings. Moreover, we introduce uniform tree languages and show that any uniform tree language is not s-recognizable. 

We will need the following obvious fact.

\begin{ob}\rm \label{ob:TU-k-rec} Both $\emptyset$ and the set $T_U^{(k)}$ are s$k$-recognizable
for every set $U$ and $k \in \nat$.
\end{ob}

In the following we give a characterization of s-recognizable tree languages in terms of recognizable tree languages and relabelings. First we prove the next lemma.

\begin{lm}\label{char-rec-lemma}\rm $\ $
\begin{enumerate}
\item For every s$k$-recognizable tree
  language $L$ we can effectively construct a $k$-bounded recognizable tree
  language $L'$ and a $k$-tree relabeling $\tau$ such that $L =
  \tau(L')$. 
\item For every  $k$-bounded recognizable tree language $L'$ and
  $k$-tree relabeling $\tau$ we can effectively construct an s$k$-recognizable tree
  language $L$ such that  $L =
  \tau(L')$.
\end{enumerate}
\end{lm}
\begin{proof} First assume that $L=L(\A)$ for some s$k$-ta $\A = (Q,U,\Phi,F,R)$. We construct the finite-state tree automaton
$\A' = (Q,\Sigma,\delta,F)$, where
\begin{itemize}
\item $\Sigma_l = \{[\varphi,l] \mid (q_1\ldots q_l,\varphi,q)\in R \text{ for some }  q_1,\ldots, q_l,q\in Q\} $, $0\le l\le k$ and
\item $\delta_{[\varphi,l]} = \{ (q_1\ldots q_l,q) \mid (q_1\ldots q_l,\varphi,q)\in R \}$.
\end{itemize}
It should be clear that $L(\A')$ is $k$-bounded. Moreover, we define the relabeling $\tau: \Sigma \to \PS(U)$ by $\tau([\varphi,l])=\seml \varphi\semr$ for every $[\varphi,l]\in\Sigma$.

We can easily prove the following statement by induction on trees: for every $\xi\in T_U\ui k$ and $q\in Q$ we have
\[ \xi \in L(\A,q) \iff \exists(\zeta \in L(\A',q)) \text{ such that } \xi \in \tau(\zeta),\]
which proves that $L(\A)=\tau(L(\A'))$.

For the proof of the other implication,  let us consider a finite-state tree automaton $\A' = (Q,\Sigma,\delta,F)$ 
such that $L(\A')$ is $k$-bounded. We may assume without loss of generality that $\maxrk(\Sigma) \le k$. Moreover,
let $\tau: \Sigma \to \PS(U)$ be a relabeling. We construct the s$k$-ta $\A = (Q,U,\Phi,\delta',F)$, where $\Phi$ and $R$ are defined as follows:
\begin{itemize}
\item $\Phi=\{ \varphi_\sigma \mid \sigma \in \Sigma\}$, where $\seml \varphi_\sigma \semr =\tau(\sigma)$ for every $\sigma \in\Sigma$, 
\item $\delta'=\{ (q_1\ldots q_l,\varphi_\sigma,q) \mid (q_1\ldots q_l,q) \in \delta_\sigma \text{ for some } l\ge 0, \sigma \in\Sigma_l\}$.
\end{itemize}
It should be clear that $L(\A)=\tau(L(\A'))$.
\end{proof}

By letting $\tau$  be the identity mapping in Lemma \ref{char-rec-lemma}(2), we obtain that each recognizable tree language is also s-recognizable.
A further consequence of Lemma \ref{char-rec-lemma} is the mentioned  characterization. 
\begin{theo}\label{th:char}
A tree language $L$ is
  s$k$-recognizable if and only if it is the image of a $k$-bounded recognizable
  tree language under a $k$-tree relabeling.
\end{theo}

Using the above characterization result, we can easily give examples of bounded tree languages that are not s-recognizable. For an infinite $U$, we call a tree language $L\subseteq T_U$ {\em uniform} if it satisfies the following conditions:
\begin{enumerate}
\item[(a)] $L$ is infinite,
\item[(b)] all trees in $L$ have the same shape, i.e., for every $\xi,\zeta \in L$, we have $\pos(\xi)=\pos(\zeta)$, and
\item[(c)] for every $\xi \in L$, there is an $a\in U$ such that $\xi(w)=a$ for every $w\in \pos(\xi)$.
\end{enumerate}
For instance, the tree language $L_2=\{a(a)\mid a\in U\}$ is uniform provided $U$ is infinite.
In particular, $\pos(\xi)=\{\varepsilon,1\}$ for every $\xi\in L_2$.
Now we can prove the following.

\begin{lm}\rm \label{non-rec-lemma} Let $L\subseteq T_U\ui k$ be a uniform tree language such that $|\pos(\xi)| > 1$ for every $\xi \in L$.
Then $L$ is not s$k$-recognizable.
\end{lm}
\begin{proof} We prove by contradiction, i.e., we assume that $L$ is s$k$-recognizable. By Lemma \ref{char-rec-lemma}(1), there is a ranked alphabet $\Sigma$,
a $k$-bounded recognizable tree language $L'\subseteq T_\Sigma$, and a  relabeling $\tau: \Sigma \to \PS(U)$ such that $L=\tau(L')$. Since $\tau$, being a $k$-tree relabeling,
preserves the shape of trees, the shape of all trees in $L'$ is the same as that of all trees in $L$. Then,  since $\Sigma$ is a finite set, $L'$ is also finite.
Finally, since $L$ is infinite, there are a tree $\zeta \in L'$,
different positions $v$ and $w$ of $\zeta$, and different labels $a, b \in U$ such that $a \in \tau(\zeta(v))$ and $b \in \tau(\zeta(w))$. Then there is a tree $\xi \in \tau(\zeta)$ such that $\xi(v)=a$ and $\xi(w)=b$, which contradicts to condition (c) for uniform tree languages.
\end{proof}

By the above lemma, for an infinite $U$, the 1-bounded tree language $L_2$ is not s-recognizable.

\subsection{Comparison with variable tree automata}

In \cite{grukupshe10} another automaton model with infinite input alphabet was introduced. It is called
variable (string) automaton. In  \cite{menrah11} this concept has been extended to 
variable tree automata over infinite alphabets (vta). The theory of
vta is different from that of sta, e.g., the class of s-recognizable
tree languages is closed under complement (cf. 
\cite[Prop. 3]{veabjo11a}) which does not hold for the class of v-recognizable
tree languages (cf. \cite[Cor. 2]{menrah11}, and  
\cite[Thm. 2]{grukupshe10}). Moreover, every sta is determinizable (cf.  
\cite[Thm. 1]{veabjo11a}), whereas not every variable (string) automata over infinite
alphabets is determinizable (cf. \cite[Sec. 4.1]{grukupshe10}). 

In this section we will compare the recognition power of sta and of vta. In order to be able to do so,  (1) we modify our sta model a bit and then (2) we recall the concepts of vta from \cite{menrah11} in a slightly
adapted form.

By a {\em ranked set} we mean a nonempty set $U$ of symbols such that with each symbol $a \in U$ an element in $\mathbb{N}$, the {\em rank of $a$},  is
  associated. For every $l
\ge 0$, we  denote by $U_{l}$ the set of all symbols of $U$ with rank $l$. 

The set of trees over a ranked set $U$ is defined in the obvious way.

An s$k$-ta ${\cal A} = (Q,U,\Phi,F,R)$ is a {\em ranked s$k$-ta (rs$k$-ta)} if
\begin{itemize}
\item $U$ is a  ranked set, and
\item $\Phi$ is a finite set of predicates (we do not require that predicates in $\Phi$ are recursive and that the emptiness problem in $\Phi$ is decidable).
   \end{itemize} 
The concepts of an rs$k$-recognizable tree language and an rs-recognizable tree language are defined in the obvious way.

Now we prepare the definition of a variable tree automaton. Let $U$ and $V$ be ranked sets.  A {\em rank preserving relabeling (r-relabeling) from $U$ to $V$} is a mapping $\tau: U \rightarrow {\cal P}(V)$ such that $\tau(U_l) \subseteq V_l$ ($l\geq 0$). Then $\tau$ extends to trees in the same way as  in case of $k$-tree relabelings (cf. Section \ref{sect:tree-trans}).
We note that $\tau(a)$ need not be recursive and $\tau(a)=\emptyset$ need not be decidable for
$a \in U$. 

Let $\Sigma$ be a ranked alphabet, $V$ an infinite ranked set,  $A$, $Z$, and $Y$ 
ranked alphabets. We say that the collection $(A,Z,Y)$ is a {\em valid partitioning
  of $\Sigma$ for $V$} if 
\begin{itemize}
\item $A=\Sigma  \cap V$, and $\Sigma_l = A_l \cup Z_l \cup Y_l$ for every $l\geq 0$, 
\item $A$, $Z$, and $Y$ are pairwise disjoint, and 
\item $|Y_l| \le 1$ for every $0 \le l \le \maxrk(\Sigma)$.
\end{itemize}
The elements of $Z$ and $Y$ are called {\em bounded variable symbols} and {\em
  free variable symbols}. 

Let $(A,Z,Y)$ be a valid partitioning
  of $\Sigma$ for $V$ and $\tau: \Sigma \rightarrow {\cal P}(V)$ an
  r-relabeling. We say that $\tau$ is {\em $(A,Z,Y)$-valid} if 
\begin{enumerate}
\item[(i)] $\tau$ is the identity on $A$,
\item[(ii)] $|\tau(z)| = 1$ for every $z \in Z$,
\item[(iii)] $\tau$ is injective on $Z$ and   $A_l \cap \tau(Z_l) = \emptyset$ for every $l \ge 0$, and 
\item[(iv)] $\tau(y) = V_l \setminus (A_l \cup \tau(Z_l))$ for
  every $l \ge 0$ and $y \in Y_l$.
\end{enumerate}
We denote the set of all $(A,Z,Y)$-valid r-relabelings  
 by $\VR(A,Z,Y)$. In Fig. \ref{fig:valid-rel} we illustrate the conditions for a valid
r-relabeling. 

\begin{figure}
\begin{center}
\includegraphics[
scale=0.5]{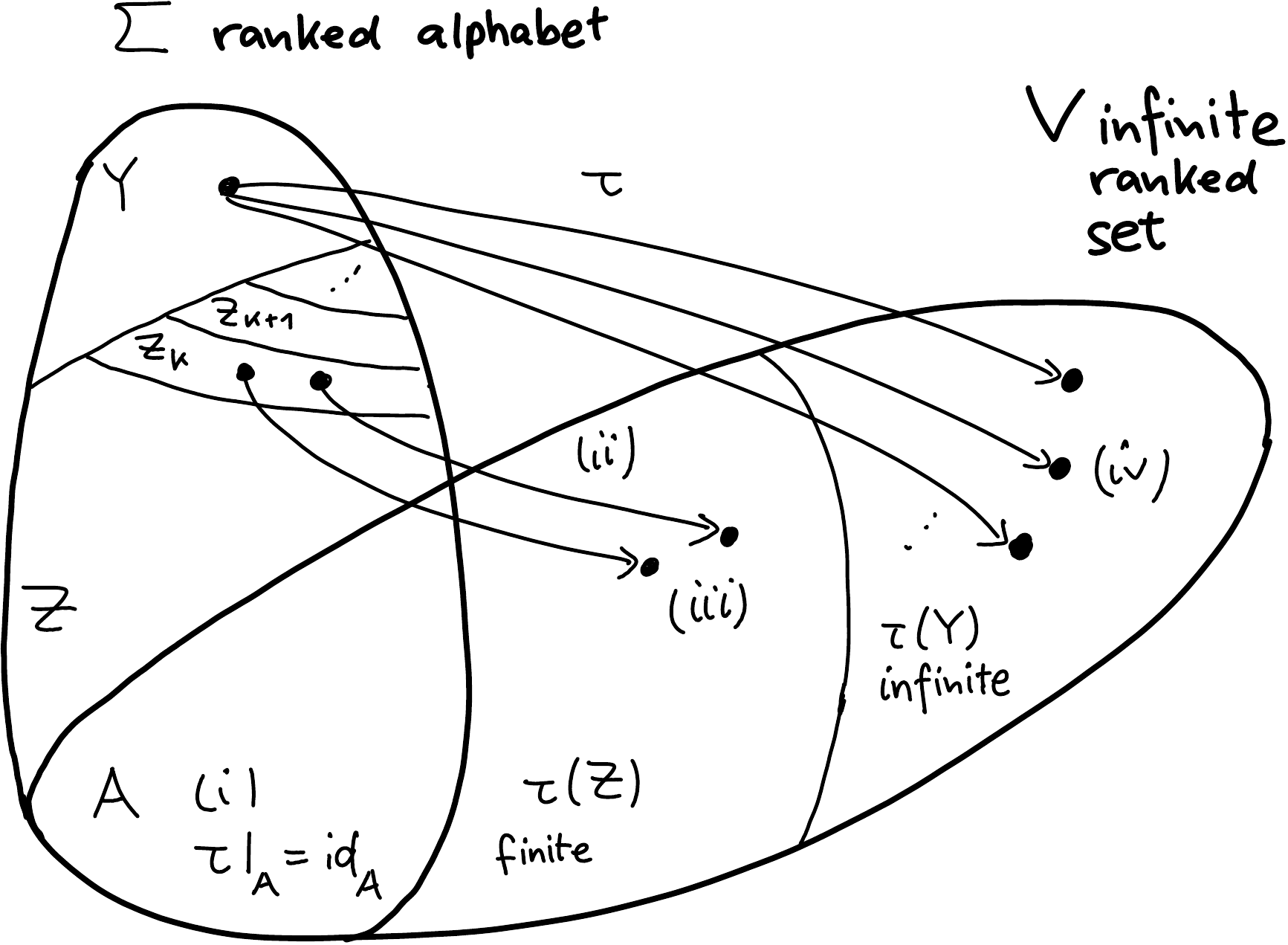}
\caption{\label{fig:valid-rel} An $(A,Z,Y)$-valid r-relabeling $\tau: \Sigma \rightarrow {\cal P}(V)$.} 
\end{center}
\end{figure}

A {\em variable tree automaton (vta)}  is a tuple ${\cal B} =
({\cal A},V,A,Z,Y)$ where 
\begin{itemize}
\item ${\cal A} = (Q,\Sigma,\delta,F)$ is a finite-state tree automaton,
\item $V$ is an infinite ranked set,
\item $(A,Z,Y)$ is a valid partitioning
  of $\Sigma$ for $V$.
\end{itemize}
The tree language {\em recognized} by ${\cal B}$ is the set 
\[
L({\cal B}) = \bigcup (\tau(L({\cal A})) \mid \tau \in \mathrm{\VR}(A,Z,Y))\enspace.
\]
We call a tree language {\em v-recognizable} if it can be recognized by a vta.

\begin{propos}\label{prop:rec-then-vrec} For every ranked alphabet $\Sigma$, every recognizable tree language $L$ over $\Sigma$ is also v-recognizable.
\end{propos}
\begin{proof} 
Let $\cal A$ be a finite-state tree automaton (with input ranked alphabet $\Sigma$) such that $L=L({\cal A})$. Moreover,
let $V$ be an arbitrary infinite ranked set such that $\Sigma \subseteq V$.
We observe that $(\Sigma,\emptyset,\emptyset)$ is a valid partitioning of $\Sigma$ for $V$, hence 
${\cal B} = ({\cal A},V,\Sigma,\emptyset,\emptyset)$ is a vta over $V$. Moreover, the only 
$(\Sigma,\emptyset,\emptyset)$-valid r-relabeling is the identity mapping over $\Sigma$. Hence we obtain that $L({\cal B})=L({\cal A})$.
\end{proof}

Next we relate v-recognizable tree languages  and s-recognizable tree languages.

\begin{propos}\label{prop:vta-sta} Let ${\cal B} = ({\cal A},V,A,Z,Y)$ be a vta and ${\cal
    A}$ have input ranked alphabet $\Sigma$. Then there is a family $(L_\tau
  \mid \tau \in \VR(A,Z,Y))$ of 
rs$k$-recognizable tree languages  over $V$ such that
  $$L({\cal B}) = \bigcup (L_\tau \mid \tau \in \VR(A,Z,Y)),$$
where $k=\maxrk(\Sigma)$.
\end{propos}
\begin{proof} We note that every $\tau \in \VR(A,Z,Y)$ is a
$k$-tree relabeling. Hence, by an easy adaptation of Lemma
\ref{char-rec-lemma}(2), we have that the tree language $\tau(L({\cal A}))$ is recognizable by an
rs$k$-ta ${\cal A}_\tau$. Hence the statement holds with $L_\tau=L({\cal A}_\tau)$.
\end{proof}

In spite of the above fact, we can prove the following statement.

\begin{theo} \label{var-vs-sym-theo}The class of v-recognizable tree languages and the class of s-recognizable tree languages are incomparable with respect to inclusion.
\end{theo}
\begin{proof}

a) We give a v-recognizable tree language and show that it is not rs-recognizable. For this, let $V=V_0 \cup V_1$ be an infinite ranked set such that $V_0 =\{c\}$, and consider the ranked alphabet $\Sigma= \Sigma_0 \cup \Sigma_1$ with $\Sigma_0 = \{c\}$ and $\Sigma_1 = \{z\}$ with $z\not\in V$.
Let $\A$ be a finite-state tree automaton
with input ranked alphabet $\Sigma$ such that $L(\A)=\{zzc\}$ (where parentheses are omitted). Now $(\{c\},\{z\},\emptyset)$ is a valid partitioning of $\Sigma$ for $V$, hence $\B=(\A,V,\{c\},\{z\},\emptyset)$ is a vta over $V$. Since every $(\{c\},\{z\},\emptyset)$-valid r-relabeling takes $z$ to an element $a\in V_1$, we have   $L(\B)=\{aac \mid a\in V_1\}$. By an easy adaptation of Lemma \ref{non-rec-lemma} we obtain that $L(\B)$ is not rs-recognizable.

b) We give an rs$1$-recognizable tree language and show that it
is not v-recognizable. 

For this, let $\Sigma = \Sigma_0 \cup \Sigma_1$ be a ranked alphabet with $\Sigma_0 = \{a\}$ and $\Sigma_1 = \{e, o\}$, and let 
$V = V_0 \cup V_1$ be an infinite ranked set with $V_0 = \{ a \}$ and $V_1 = \mathbb{N}$ .
Consider the recognizable tree language $$L= \{ a, oa, eoa, oeoa, eoeoa, \ldots\}$$ over $\Sigma$ (where parentheses are omitted),
and the r-relabeling $\overline{\tau}$ defined by

$\overline{\tau}(a) = a$, $\overline{\tau}(e)=$ set of all even numbers, and $\overline{\tau}(o)=$ set of all odd
numbers.

By the adaptation of Lemma \ref{char-rec-lemma}(2) to rs$k$-ta, we obtain that the tree language
$\overline{\tau}(L)$ can be recognized by an rs$1$-ta. Roughly speaking, $\overline{\tau}(L)$ consists of all sequences of the form $n_k\ldots n_2n_1a$, where $k\geq 0$, $n_i$ is an odd number if $i$ is odd and an even number otherwise.
We show by contradiction that $\overline{\tau}(L)$ cannot be recognized by any vta.

For this, assume that there is a vta ${\cal B} = ({\cal A},V,A,Z,Y)$, where the input alphabet of $\A$ is $\Sigma=A\cup Z\cup Y$ such that $L(\B)=\overline{\tau}(L)$. We may assume without loss of generality that $Y=\emptyset$, which can be seen as follows. Assume that $Y=\{y\}$, and that there is a tree $\xi \in L(\A)$ such that $y$ occurs in $\xi$ at the position $1^i$ (see Section \ref{sect:trees} for the definition of a position). 
Moreover, let $\tau: \Sigma \rightarrow {\cal P}(V)$ be an $(A,Z,Y)$-valid r-relabeling. Since $\tau(y)$ contains both even and odd numbers, there are trees $\zeta$ and $\zeta'$ in the set $\tau(\xi)$ such that at the position $1^i$ of $\zeta$ and $\zeta'$ there is an odd number and an even number, respectively. On the other hand, $\tau(\xi)\subseteq L(\B)$, which is a contradiction because at the  position $1^i$ of every tree in $L(\B)$ there is either an odd number or an even number (depending on whether $i$ is odd or even).

Hence $Y=\emptyset$. Now assume that $|A\cup Z| = m$. Then every tree in $\xi \in L(\A)$ consists of at most $m$ different symbols. Moreover, by the definition of the $(A,Z,Y)$-valid r-relabeling,
for every $\tau \in \mathrm{\VR}(A,Z,Y)$, each tree in $\tau(\xi)$ consists of $m$ different symbols. It means, each tree
in $L(\B)$ consists of $m$ different symbols, which  contradicts to the much more flexible form of trees in $\overline{\tau}(L)$.
\end{proof}

\section{Symbolic regular tree grammars} 

In this section we introduce symbolic regular tree grammars and show that they are semantically equivalent to sta. 

\begin{df} \rm A \emph{symbolic $k$-bounded  regular tree grammar} (s$k$-rtg) is a
  tuple $\G = (Q,U,\Phi,q_0,R)$, where
\begin{itemize}
\item $Q$ is a finite set (states\footnote{In classical regular tree
    grammars these elements are called nonterminals.}),
\item $(U,\Phi)$ is a label structure,
\item $q_0 \in Q$ (initial state), and 
\item $R$ is a finite set of rules of the form $q \rightarrow u$ where $q \in Q$ and $u \in T_{\BC(\Phi)}\ui k(Q)$. 
\end{itemize}
\end{df}
By a {\em symbolic regular tree grammar} (srtg) we mean an s$k$-rtg for some $k\in \nat$.

The s$k$-rtg $\G = (Q,U,\Phi,q_0,R)$ induces the derivation
relation $\Rightarrow_\G \subseteq T_U\ui k(Q) \times T_U\ui k(Q)$ defined by
$\xi_1 \Rightarrow_\G \xi_2$ iff there is a position $w \in \pos_q(\xi_1)$ and a rule $q \rightarrow u$ in $R$,  such that $\xi_2 = \xi_1[u']_w$, where $u'$ is
obtained from $u$ by replacing every occurrence of $\varphi\in \BC(\Phi)$ by some $a \in \seml \varphi\semr$. (The condition that all predicates in $\Phi$ are recursive makes the relation $\Rightarrow_\G$ recursive.)

The $k$-bounded tree language $L(\G,q)$ generated by $\G$ from a state $q\in Q$ is the set 
\[
L(\G,q) = \{\xi \in T_U\ui k \mid q \Rightarrow_\G^* \xi\}\enspace.
\] 
The \emph{tree language generated by $\G$},
denoted by $L(\G)$, is the set $L(\G,q_0)$. A tree language $L \subseteq
T_U\ui k$ is called \emph{symbolically $k$-regular} (for short: s$k$-regular) if there is an s$k$-rtg $\G$ such that $L = L(\G)$. Moreover, a tree language is s-regular if it is s$k$-regular for some $k\in \nat$.

Two s$k$-rtg $\G_1$ and $\G_2$ are \emph{equivalent} if $L(\G_1)= L(\G_2)$.

In the following we give a characterization of s-regular tree
languages in terms of regular tree languages and relabelings.

\

\begin{lm}\label{lm:srtg-alph}\rm $\ $
\begin{enumerate}
\item For every s$k$-regular tree
  language $L$ we can effectively construct a $k$-bounded regular tree
  language $L'$ and a $k$-tree relabeling $\tau$ such that $L =
  \tau(L')$. 
\item For every  $k$-bounded regular tree language $L'$ and
  $k$-tree relabeling $\tau$ we can effectively construct an
  s$k$-regular tree language $L$ such that  $L =
  \tau(L')$.
\end{enumerate}
\end{lm}
\begin{proof} First let $L = L(\G)$ for some s$k$-rtg $\G = (Q,U,\Phi,q_0,R)$. We construct the regular tree grammar $\G' = (Q,\Sigma,q_0,R')$ as follows.
\begin{itemize}
\item For every $l\le k$, let $$\Sigma_l = \{ [\varphi,l] \mid \exists (q \rightarrow u) \in R, 
w \in \pos_{\BC(\Phi)}(u):  u(w) = \varphi \text{ and } \rk_w(u) = l\},$$ 
\item and let $R'$ be the set of all rules $q \rightarrow u'$ such
  that there is a rule $q \rightarrow u$ in $R$ and $u'$ is obtained from $u$ as follows:
for every $w\in \pos(u)_{\BC(\Phi)}$, we replace $u(w)$ by $[u(w),\rk_w(u)]$.
\end{itemize}
It is obvious that $L(\G')$ is $k$-bounded. Moreover, we let the relabeling $\tau: \Sigma \rightarrow {\cal P}(U)$ be defined by $\tau([\varphi,l]) = \seml \varphi \semr$
for every $0\le l\le k$ and $[\varphi,l] \in \Sigma_l$.

We can prove the following statement by tree induction: 

\begin{quote}
for every $\zeta \in T_U\ui k$ and $q \in Q$:\\
$q \Rightarrow_{\cal G}^* \zeta$ \;  iff \; 
there is a $\xi \in T_\Sigma$: $q \Rightarrow_{\G'}^* \xi$ and $\zeta \in \tau(\xi)$.
\end{quote} 
Then $L(\G) = \tau(L(\G'))$. 

For the proof of Statement 2,  let us consider a regular tree grammar  $\G' = (Q,\Sigma,q_0,R)$ such that $L(\G')$ is $k$-bounded and a relabeling $\tau: \Sigma \to \PS(U)$. We may assume without loss of generality that  $\maxrk(\Sigma)\le k$. We construct the s$k$-rtg $\G = (Q,U,\Phi,q_0,R')$, where $\Phi$ and $R'$ are defined as follows:
\begin{itemize}
\item $\Phi=\{ \varphi_\sigma \mid \sigma \in \Sigma\}$, where $\seml \varphi_\sigma \semr =\tau(\sigma)$ for every $\sigma \in\Sigma$, 
\item $R'$:  if $q \rightarrow u$ is in $R$, then $q \rightarrow u'$ is in $R'$ where $u'$ is obtained from $u$ by replacing every $\sigma$ by $\varphi_\sigma$. 
\end{itemize}
It should be that $L(\G)=\tau(L(\G'))$.
\end{proof}

It follows from Lemma \ref{lm:srtg-alph}(2) that each regular tree language is also s-regular. We obtain this by letting $\tau$  be the identity mapping.
As another consequence of Lemma  \ref{lm:srtg-alph}, we obtain the
following characterization result. 

\begin{theo} \label{th:char-reg}
A tree language $L$ is
  s$k$-regular if and only if it is the image of a $k$-bounded regular tree language under a $k$-tree relabeling.
\end{theo}

We can also show that s-recognizable tree languages are the same as s-regular tree languages. 

\begin{theo}\label{th:rec=reg} A tree language is s-recognizable if and only if it is s-regular.
\end{theo}
\begin{proof} It follows directly from Lemmas \ref{char-rec-lemma}
  and \ref{lm:srtg-alph} and the fact that a tree language is recognizable if and only if it can be generated by a regular tree grammar (cf. e.g. Theorem 3.6 in Chapter II of
  \cite{gecste84}).
\end{proof}

In the rest of this section we show some useful transformations on s$k$-rtg which preserve the generated tree language.
For this, we need some preparation.

Let $\G = (Q,U,\Phi,q_0,R)$ be an s$k$-rtg. A rule $q\to u$ in $R$ is {\em feasible} if, for every predicate $\varphi$ which occurs in $u$, we have $\seml \varphi\semr \neq \emptyset$, and we call
$\G$ {\em clean} if all its rules are feasible.  It is obvious that rules which are non-feasible cannot be used in any valuable derivations. 
Hence, they can be dropped from $R$ without any effect on the generated tree language $L(\G)$. Moreover, it is decidable whether a rule is feasible or not due to the fact that the emptiness of predicates in $\Phi$ is  decidable. Summarizing up, for every s$k$-rtg  we can construct an equivalent one, which is clean.

The s$k$-rtg $\G$ is in \emph{normal form} if every rule
has the form $q\to \varphi(q_1,\ldots,q_l)$ for some $l\le k$, $\varphi \in \BC(\Phi)$, and
$q_1,\ldots,q_l \in Q$. A state $q\in Q$ is {\em reachable} if there is a tree $\xi \in T_U\ui k(Q)$ such that  $q_0 \Rightarrow_\G^* \xi$ and $q$ occurs in $\xi$. Moreover, the state  $q$ is {\em productive}, if $L(\G,q)\ne \emptyset$. Finally, $\G$ is {\em reduced} if all its states
are reachable and productive. We can prove the following result.

\begin{lm}\rm \label{lm:srtg-normal-form} For every s$k$-rtg  there is an equivalent reduced s$k$-rtg which is in
  normal form.
\end{lm} 
\begin{proof} Let $\G = (Q,U,\Phi,q_0,R)$ be an s$k$-rtg. We may assume that $\G$ is clean.
By Lemma
  \ref{lm:srtg-alph}(1), there is a regular tree grammar $\G'$ over some ranked alphabet $\Sigma$ 
such that $L(\G)$ is $k$-bounded, and there is a
  relabeling $\tau: \Sigma \rightarrow {\cal P}(U)$ such that $L(\G) =
  \tau(L(\G'))$. Since $\G$ is clean, $\tau(\sigma)\neq  \emptyset$ for every $\sigma \in \Sigma$ (see  the proof of that lemma).

Then we transform $\G'$ into an equivalent regular tree grammar
  $\G''$ which is reduced and is in normal form using the transformations in \cite[Prop. 2.1.3, 2.1.4]{comdaugiljaclugtistom97}.
Note that $L(\G'')$ is $k$-bounded and $L(\G) =  \tau(L(\G''))$.

Finally, we follow the proof of Lemma  \ref{lm:srtg-alph}(2) to construct an s$k$-rtg $\overline{\G}$ from $\G''$ and $\tau$ such that $L(\overline{\G}) =  \tau(L(\G''))$. Then $\overline{\G}$ is clean due to the above condition on $\tau$. Moreover,
a direct inspection of that construction shows that $\overline{\G}$ is reduced and is in normal form.
\end{proof}

\section{Symbolic tree transducers}

In this section we formalize our adaptation of the concept of a symbolic tree transducer from \cite{veabjo11a,veabjo11b}. Then we show some basic properties, relate symbolic tree transducers to classical top-down tree transducers  \cite{tha70,rou70,eng75}, and compare our model with the original one. Finally, we prove a composition result for symbolic tree transducers.

\subsection{Definition of stt}\label{sect:stt-def}

For every finite set $Q$ and $l \in \nat$, we let $Q(X_l) =
\{q(x_i) \mid q \in Q, x_i \in X_l\}$. 

We denote by $\F(U \rightarrow V)$ the set of all unary computable functions from $U$ to $V$. Moreover, for every tree
$u \in T_{\F(U \rightarrow V)}(Y)$ and $a\in U$, we denote by $u(a)$ the tree which is obtained by replacing every function $f$ in $u$ by the value $f(a)\in V$. Hence we have that $u(a)\in T_{V}(Y)$.

\begin{df}\rm  Let $k \in \nat$. A \emph{symbolic $k$-bounded  tree transducer} (s$k$-tt) is a tuple
  $\M = (Q,U,\Phi,V,q_0,R)$, where
\begin{itemize}
\item $Q$ is a finite set (states),
\item $(U,\Phi)$ is a label structure (input label structure) and $V$ is a set (output labels),
\item $q_0 \in Q$ (initial state), and 
\item $R$ is a finite set of rules of the form $q(\varphi(x_1,\ldots,x_l)) \to u$ 
where $q \in Q$, $\varphi \in \BC(\Phi)$, $0 \le  l\le k$, and $u \in T_{\F(U \rightarrow V)}\ui k(Q(X_l))$. 
\end{itemize}
\end{df}
Clearly, every s$k$-tt is an s$(k+1)$-tt. By an stt we mean an s$k$-tt for some $k \in \nat$.  

For a rule $\rho = q(\varphi(x_1,\ldots,x_l)) \to u$, we call the pair $(q,l)$ the \emph{left-hand side state-rank pair}, $\varphi$ the \emph{guard}, and $u$ the
\emph{right-hand side} of $\rho$, and denote them  by $\lhs(\rho)$,  $\grd(\rho)$, and $\rhs(\rho)$, respectively. 

We say that the stt $\M$ is \emph{linear} (resp. \emph{nondeleting}) if, for each rule $\rho$ as above, its right-hand side contains at most (resp. at
least) one occurrence of $x_i$ for every $1 \le i \le l$.

Moreover,  $\M$ is \emph{deterministic} if, for any two different rules $\rho_1$ and $\rho_2$ in $R$, the condition $\lhs(\rho_1)=\lhs(\rho_2)$ entails that $\seml \grd(\rho_1) \semr\cap \seml \grd(\rho_2) \semr=\emptyset$. Finally, $\M$ is {\em total} if for every $q\in Q$ and $0\le l \le k$, we have
\[\seml \bigvee_{\substack{\rho\in R \\ \lhs(\rho)=(q,l)}}
\grd(\rho)\semr =U .\]
We note that, as for sta, no s$k$-tt is a total $s(k+1)$-tt.

Next we define the semantics of an s$k$-tt $\M = (Q,U,\Phi,V,q_0,R)$. 
We define the \emph{derivation relation of $\M$}, denoted by $\Rightarrow_{\M}$, to be the smallest binary relation $\Rightarrow_{\M} \subseteq T_V(Q(T_U)) \times T_V(Q(T_U))$ such that for every $\xi_1,\xi_2 \in T_V(Q(T_U))$:

$\xi_1 \Rightarrow_{\M} \xi_2$ iff there is a position $w \in \pos(\xi_1)$ and a rule $q(\varphi(x_1,\ldots,x_l)) \to u$ in $R$,  such that

\begin{itemize}
\item $\xi_1|_w = q(a(\zeta_1,\ldots,\zeta_l))$ for some $a \in \seml \varphi \semr$ and $\zeta_1,\ldots,\zeta_l \in T_U^{(k)}$, and
\item $\xi_2=\xi_1[u']_w$, where $u'$ is obtained from $u(a)$ by replacing every index $p(x_i) \in Q(X_l)$ by $p(\zeta_i)$. 
\end{itemize} 
The conditions that all predicates in $\Phi$ are recursive and that all functions in the right-hand side of the rules are computable make the relation $\Rightarrow_{\M}$ recursive. Sometimes, we drop $\M$ from $\Rightarrow_{\M}$. 

Let $q \in Q$, $\xi \in T_U$, and $\zeta \in T_V$. We can show by induction on $\xi$ that
if $\xi \in T_U\ui k$ and $q(\xi)\Rightarrow_{\M}^* \zeta$ holds, then also  $\zeta \in T_V\ui k$.
The \emph{$q$-tree transformation} computed by $\M$, denoted by $\M_q$, is the relation
\[
\M_q = \{(\xi,\zeta) \in T_U^{(k)} \times T_V^{(k)} \mid q(\xi) \Rightarrow_{\M}^* \zeta\}\enspace.
\]
The \emph{tree transformation} computed by $\M$, also denoted by $\M$, is defined by $\M = \M_{q_0}$.
 The class of
tree transformations computed by s$k$-tt (resp. linear, nondeleting, deterministic, and total,  s$k$-tt) is denoted by $\STT\ui k$ (resp. $\lSTT\ui k$,
$\nSTT\ui k$, $\dSTT\ui k$, and $\tSTT\ui k$). These restrictions  can
be combined in the usual way, for instance, we will denote by
$\lnSTT\ui k$ the class of tree transformations computed by linear and
nondeleting  s$k$-tt.

A deterministic s$k$-tt (total s$k$-tt) transforms every input tree
into at most one (at least one) output tree.

\begin{lm}\rm \label{dt-lemma-stt} If $\M$ is a deterministic (resp. total) s$k$-tt, then  we have $|\M_q(\xi)|\le 1$  (resp. $|\M_q(\xi)|\ge 1$) for every $q \in Q$ and $\xi \in T_U\ui k$.
\end{lm}

\begin{ex}\rm\label{stt-example}\rm We consider the s2-tt ${\cal M} = (Q,U,\Phi,U,q,R)$
  with $Q = \{q\}$, $U = \mathbb{N}$, and $\Phi = \{\mathrm{div}(2),
  \mathrm{div}(3)\}$ with $\seml \mathrm{div}(i) \semr$ is the set of
  all non-negative integers which are divisible by $i$. Moreover, $R$
  has the following rules:
\[
\begin{array}{lrcl}
\rho_1: & q\big([\mathrm{div}(2)\wedge \mathrm{div}(3)](x_1,x_2)\big) & \rightarrow &
[:6](q(x_1),q(x_1))\\
\rho_2: & q(\top(x_1,x_2)) & \rightarrow & [\mathrm{id}](q(x_1),q(x_2))\\
\rho_3: & q(\top) & \rightarrow & [\mathrm{id}]\\
\end{array}
\]
where the unary functions $[:6]$ and $\mathrm{id}$ perform division by
6 and the identity, respectively. Note that ${\cal M}$ is
not deterministic, because $\lhs(\rho_1) =
\lhs(\rho_2) = (q,2)$ and 
\[
\seml \grd(\rho_1) \semr \cap \seml \grd(\rho_1)\semr = 
\seml \mathrm{div}(2)\wedge \mathrm{div}(3) \semr \cap \seml \top
\semr = \seml \mathrm{div}(2)\wedge \mathrm{div}(3) \semr \not= \emptyset\enspace.
\] 

Also note that ${\cal M}$ is not total, because for $l=1$ we have:

\[
\seml \bigvee_{\substack{\rho\in R \\ \lhs(\rho)=(q,1)}}
\grd(\rho)\semr = 
\seml \bigvee_{\rho\in \emptyset} \grd(\rho)\semr = 
\seml \bot \semr = \emptyset \not= U\enspace.
\]
Also ${\cal M}$ is neither linear nor nondeleting, because of rule $\rho_1$.

On the input tree $\xi = 6(12(4,6),7)$ the s2-tt ${\cal M}$ can perform
  the following derivation:

\[
\begin{array}{cl}
    & q(6(12(4,6),7)) \\
\Rightarrow & 1(q(12(4,6)), q(12(4,6))) \\
\Rightarrow^2 & 1(2(q(4),q(4)), 12(q(4),q(6))) \\
\Rightarrow^4 & 1(2(4,4), 12(4,6))
\end{array}
\]

The s2-tt ${\cal M}$ transforms a binary tree $\xi$ in the following
way. At each position $w$, ${\cal M}$ can reproduce the label $\xi(w)$ of this
position and recursively transforms the subtrees (using rules $\rho_2$
and $\rho_3$). If $\xi(w)$ is
divisible by $6$, then, additionally (using rule $(\rho_1)$), ${\cal
  M}$ can divide it by 6, delete the second subtree, and process two copies of the first
subtree independently.  
\end{ex}

Next we show that stt generalize (classical) top-down tree transducers. For every $b\in V$, we denote
by $c_b$ the constant function in $\F(U \rightarrow V)$ defined by $c_b(a)=b$ for every $a\in U$.
An s$k$-tt  $\M = (Q,U,\Phi,V,q_0,R)$ is \emph{alphabetic} if 
\begin{itemize}
\item $U$ and $V$ are ranked alphabets such that $\maxrk(U), \maxrk(V)  \le k$,
\item $\Phi = \{\varphi^\sigma \mid \sigma \in U\}$ where $\seml \varphi^\sigma\semr = \{\sigma\}$,

\item each  rule in
  $R$ has the form $q(\varphi^\sigma(x_1,\ldots,x_l)){\rightarrow}u$, where
\begin{enumerate}
\item[-] $l = \rk_U(\sigma)$, and
\item[-] for every $w\in (\pos(u) \setminus \pos_{Q(X_l)}(u))$ we have $u(w)=c_b$  and $\rk_u(w)= \rk_V(b)$ for some $b\in V$.
\end{enumerate}
\end{itemize}
We call predicates of the form $\varphi^\sigma$ {\em alphabetic}.

Let $\M = (Q,\Sigma,\Phi,\Delta, q_0,R)$ be an alphabetic s$k$-tt with rank mappings $\rk_{\Sigma}$ and $\rk_{\Delta}$. Let $\N = (Q,\Sigma,\Delta,q_0.R')$ be a top-down tree transducer with the same rank mappings. Then we say that $\M$ and $\N$ are \emph{related} if 
\[
q(\varphi^\sigma(x_1,\ldots,x_l)) \to u \in R \text{ iff }  q(\sigma(x_1,\ldots,x_l)) \to u' \in R'\enspace,
\]
where we obtain $u'$ from $u$ by replacing $c_\delta$ by $\delta$ for every $\delta \in \Delta$.

For every alphabetic s$k$-tt $\M$ we can construct a related top-down tree transducer $\N$  and vice versa. Moreover, it is easy to see that if $\M$ and $\N$ are related, then the tree transformations computed by $\M$ any by $\N$ are the same. Hence we obtain the following result.

\begin{ob}\rm The class of tree transformations computed by alphabetic stt is the same as the class of top-down tree transformations.
\end{ob}

Recall that $\iota_U$ is the identity mapping on $U$. Let $\A = (Q,U,\Phi,F,R)$ be an s$k$-ta. We introduce the s$k$-tt $\A_{=} = (Q,U,\Phi,U,F,R_{=})$, where
$$ R_{=} = \{ q(\varphi(x_1,\ldots,x_l)){\rightarrow} \iota_U(q_1(x_{1}), \ldots, q_l(x_{l}))
\mid (q_1\ldots q_l,\varphi,q) \in R \}.$$

We will need the following fact.

\begin{lm}\rm \label{ln-lemma} For every s$k$-recognizable tree language $L$, there is a linear and nondeleting s$k$-tt $\N$ such that $\N = \iota_{L}$.
\end{lm} 
\begin{proof} Let $L=L(\A)$ for some s$k$-ta $\A$. Then $\N=\A_{=}$ is appropriate.
\end{proof}

Finally, we want to compare our model with the original one from \cite{veabjo11b}. Each rule of their symbolic tree transducer has either
of the following two forms:
\begin{enumerate}
\item[(a)] $q(\varepsilon) \rightarrow e$ or
\item[(b)] $q(f(x,y_1,\ldots,y_k)) \stackrel{\varphi[x]}{\longrightarrow} u[x,q_1(y_1),\ldots,q_k(y_k)]$
\end{enumerate}
where $\varepsilon$ is the only nullary constructor for trees (more precisely, for the empty tree)  and $f$ is the only non-nullary constructor for trees. Since in our
approach we have neither the empty tree nor the constructor
$\varepsilon$, there are no rules in our definition of symbolic tree
transducers which correspond to rules of type (a). Also the constructor
$\varepsilon$ does not occur in the right-hand side of rules of type
(b). Then, in our approach, a rule of type (b) looks as follows:
\[
q(\varphi(y_1,\ldots,y_k)) \rightarrow \psi(u)
\]
where the transformation $\psi$ is defined inductively on its argument as follows:
\begin{itemize}
\item $\psi(f(p,u_1,\ldots,u_l)) = (\lambda x. p)\big(
  \psi(u_1),\ldots,\psi(u_l)\big)$, and
\item $\psi(q_i(y_i)) = q_i(y_i)$.
\end{itemize}
That is, $\psi$ applies the constructor $f$, replaces an expression
$p$ (in which the variable $x$ occurs) by the unary function $\lambda
x.p$, and recursively calls itself on the subterms $u_1,\ldots,u_l$.

\subsection{Composition results concerning stt}

In \cite{veabjo11b}, among others, composition properties of symbolic
tree transformations are considered. Their main result is Theorem 1
which, in its first statement, says that tree transformations computed by  stt are
closed under composition. For this they give the following  proof: ``The first statement can be shown along the lines of
the proof of compositionality of TOP [15, Theorem 3.39].'' where
``[15]'' is \cite{fulvog98} in the current paper. Unfortunately, the
mentioned proof of \cite{fulvog98} is not applicable, because there
the authors only consider total and deterministic top-down tree
transducers. 

Moreover, also on the  semantics level there is a
deficiency. In Section 4.1 they claim the following:
\begin{quote}$(\dagger)$ For  two arbitrary stt ${\cal M}$ and ${\cal N}$,
 the composition algorithm delivers an stt which computes the
composition ${\cal M} \circ {\cal N}$. 
\end{quote}
However, this is not true, which can be seen as follows. Let us apply
their composition algorithm  to two alphabetic s$k$-tt ${\cal M}$ and ${\cal N}$, then the
resulting s$k$-tt is also alphabetic (by Observation \ref{ob:alph-stt})  and, due to their claim, it
computes ${\cal M} \circ {\cal N}$. Since alphabetic s$k$-tt correspond to top-down tree
transducers it  means that the class of all top-down tree
transformations is closed under
composition.   However, it is not, due to the counter examples given
in \cite[p 267.]{rou70} (cf. also \cite{tha70,eng75}).

  So, the proof of the first statement
is insufficient. We even conjecture that this statement is
wrong, i.e., $\STT^{(k)}$ is not closed under composition.

In this section we prove a weaker version of  claim $(\dagger)$ which
only holds for particular stt ${\cal M}$ and ${\cal N}$, cf. Theorem
\ref{comp-lemma}. In fact, we  generalize the composition theorem
\cite[Thm. 1]{bak79} for top-down tree transducers to symbolic tree transducers.

For this, we use the composition algorithm of \cite{veabjo11b} which results
in the syntactic composition ${\cal M}; {\cal N}$, and we show that in certain cases the stt $\M;\N$ computes the relation $\M\circ \N$. In the following, we recall the composition algorithm of
\cite{veabjo11b} in our formal setting. We note that this composition
algorithm generalizes the (syntactic) composition of top-down tree
transducers as presented in the definition before Theorem 1 of \cite{bak79}.

Let $f\in \F(U\rightarrow V)$ and $v\in T_{\F(V\rightarrow W)}(Y)$. We
denote by $f \circ v$ the tree obtained from $v$ by replacing every occurrence of a function $g \in \F(V\rightarrow W)$ by the function $f \circ g \in \F(U\rightarrow W)$. Of course
$f \circ v \in T_{\F(U\rightarrow W)}(Y)$.

In the following it will be useful to show the occurrences of objects of the form $q(x_i)$ in the right-hand side of rules of an stt explicitly. Therefore sometimes
we write an arbitrary element of  $T_{\F(U \rightarrow V)}\ui k(Q(X_l))$ in the form $u[q_1(x_{i_1}),\ldots,q_m(x_{i_m})]$, where $m\geq 0$, $u \in C_{\F(U \rightarrow V)}\ui k(X_m)$, $q_1,\ldots,q_m \in Q$, and $1\le i_1,\ldots,i_m\le l$.

We define the syntactic composition $\M;\N$ of two s$k$-tt $\M$ and $\N$ by applying $\N$ to the right-hand side of rules of $\M$. However, we can do it only symbolically because such a right-hand side is built up from functions and not from labels. In fact, we define a symbolic version of the derivation relation $\Rightarrow_{\N}$, denoted by $\stackrel{s}{\Rightarrow_{\N}}$ which processes trees over functions. Besides, the rewrite relation $\stackrel{s}{\Rightarrow_{\N}}$ also deals with objects of the form $q(x_i)$ in its input trees. Moreover, we have to collect the Boolean combinations which are encountered during the transformation of a right-hand side.

Formally, let $\M = (Q,U,\Phi_1,V,q_{0},R_1)$ and 
$\N = (P,V,\Phi_2,W,p_{0},R_2)$ be two s$k$-tt and
$\Phi = \Phi_1 \cup \{f \circ \psi \mid f \in \F(U\rightarrow V), \psi \in \Phi_2\}$. First, we define the binary relation $\stackrel{s}{\Rightarrow_{\N}}$ over the set
\[
 \BC\left(\Phi\right)\times
T_\Sigma\Big(P\big(T_\Delta\left(Q\left(X_l\right)\right)\big)\cup (P\times Q)\left(X_l\right)\Big)
\]
where $\Sigma = \F(U\rightarrow W)$ and
$\Delta = \F(U\rightarrow V)$ (cf. Fig. \ref{fig:synt-comp-derivation}).  For every 
\[(\theta,t), (\theta',t') \in \BC\left(\Phi\right)\times
T_\Sigma\Big(P\big(T_\Delta\left(Q\left(X_l\right)\right)\big)\cup (P\times Q)\left(X_l\right)\Big)\] we have 

$(\theta,t) \stackrel{s}{\Rightarrow_{\N}} (\theta',t')$  iff one of the following two conditions hold:

(i) there is a position $w \in \pos(t)$ such that 
\begin{itemize}
\item $t|_w = p(q(x_i))$ for some $p \in P$, $q \in Q$, and $x_i \in X_l$,
\item $t' = t[\langle p,q\rangle(x_i)]_w$, and 
\item $\theta' = \theta$, or
\end{itemize}

(ii) there is a position $w \in \pos(t)$ and a rule $p(\psi(x_1,\ldots,x_l))\rightarrow v$ in $R_2$ such that 
\begin{itemize}
\item $t|_w = p(f(t_1,\ldots,t_l))$ for some $p \in P$, $f \in \F(U\rightarrow V)$,  and $t_1,\ldots,t_l \in T_{\F(U\rightarrow V)}(Q(X_l))$,
\item $t' = t[v']_w$ where $v'$ is obtained from $f \circ v$ by  replacing every $\overline{p}(x_i) \in Q(X_l)$ by $\overline{p}(t_i)$, and
\item $\theta' = \theta \wedge f \circ \psi$. 
\end{itemize}

\begin{figure}
\begin{center}
\includegraphics[
scale=0.5]{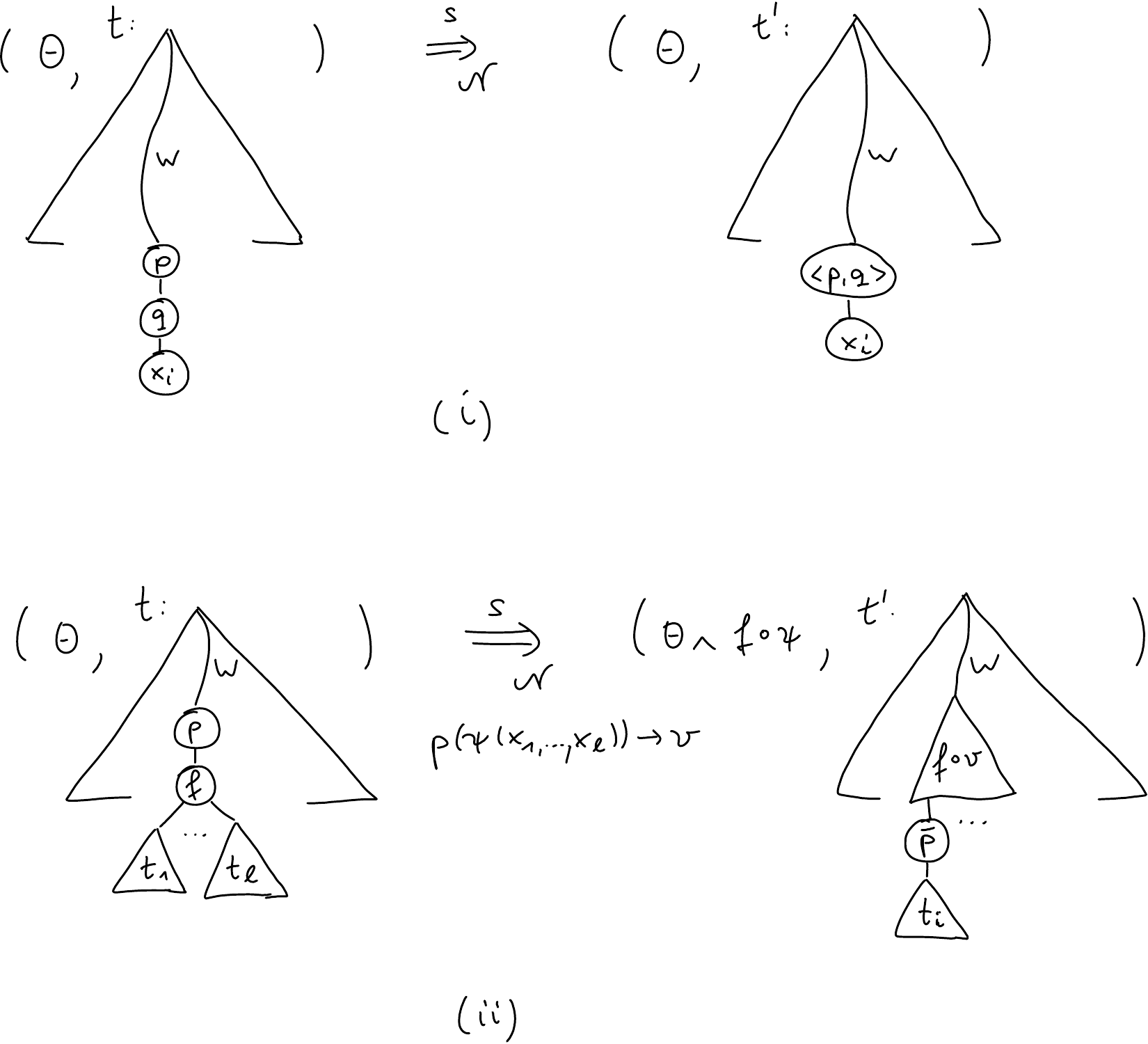}
\caption{\label{fig:synt-comp-derivation} The derivation relation $\stackrel{s}{\Rightarrow_{\N}}$.} 
\end{center}
\end{figure}

The following statement follows from the definition of the relation $\stackrel{s}{\Rightarrow_{\N}}$.

\begin{lm}\rm(Lift lemma.)\label{lift-lemma} If $(\varphi,t) \,
  (\stackrel{s}{\Rightarrow_{\N}})^* \, (\theta,t')$, then, for every $
  a\in \seml \theta\semr$ we have $ a\in \seml\varphi \semr$
and $t(a) \Rightarrow_{\N'}^* t'(a)$, where $\Rightarrow_{\N'}$ is the extension of $\Rightarrow_{\N}$ to the set 
\[
 T_W\Big(P\big(T_V\left(Q\left(X_l\right)\right)\big)\cup (P\times Q)\left(X_l\right)\Big),
\]
which we obtain by adding the rules $p(q(x_i)) \rightarrow \langle
p,q\rangle(x_i)$ to $R_2$ for every $1 \le i \le k$ (cf. p.195 of \cite{bak79}).
\end{lm}

Second, we construct the s$k$-tt $\M;\N = (P \times Q, U, \Phi,W,\langle p_{0},q_{0}\rangle,R)$,
called the {\em syntactic composition of  $\M$ and $\N$}, where 
the set $R$ of rules is defined as follows. If
\begin{equation} \label{rule-1}
q(\varphi(x_1,\ldots,x_l)) \rightarrow u[q_1(x_{i_1}),\ldots,q_m(x_{i_m})]
\end{equation}
is a rule in $R_1$, and for some $p \in P$ and $v \in C^{(k)}_{{\cal
    F}(U \rightarrow W)}(X_n)$ we have 
\begin{align}\label{derivation-2}
&\big(\varphi,p(u[q_1(x_{i_1}),\ldots,q_m(x_{i_m})])\big) \,
(\stackrel{s}{\Rightarrow_{\N}})^* \, \big(\theta,v[\langle p_1,
q_{j_1}\rangle (x_{i_{j_1}}),\ldots,\langle p_n, q_{j_n}\rangle
(x_{i_{j_n}})]\big)\\ 
&\text{ and } \seml \theta \semr \not= \emptyset \notag
\end{align}
then let the rule
\begin{equation}\label{rule-big}
\langle p,q\rangle(\theta(x_1,\ldots,x_l)){\rightarrow} v[\langle p_1, q_{j_1}\rangle (x_{i_{j_1}}),\ldots,\langle p_n, q_{j_n}\rangle (x_{i_{j_n}})]
\end{equation}
be in $R$. Note that
\[\{i_{j_1},\ldots,i_{j_n}\}\subseteq \{i_1,\ldots,i_m\} \subseteq \{1,\ldots,l\}.\]
We also note that syntactic composition preserves the properties
linear, nondeleting, total, and deterministic. For instance, if both
$\M$ and $\N$ are linear, then $\M;\N$ is also linear.

\begin{ob}\rm \label{ob:alph-stt} The syntactic composition of two
  alphabetic s$k$-tt is an alphabetic s$k$-tt.
\end{ob}
\begin{proof} Let us assume that ${\cal M}$ and ${\cal N}$ are
alphabetic. Then $\varphi$ in \eqref{derivation-2} is an alphabetic
predicate. Moreover, we observe that if in (i) of the definition of
$\stackrel{s}{\Rightarrow}_{\cal N}$, the predicate $\theta$ is alphabetic, then also
$\theta'$ is alphabetic; moreover, if  in (ii) of this definition 
$\varphi$ and $\psi$ are
alphabetic and $f$ is a constant function, then either $\seml \theta'
\semr = \emptyset$ or $\theta' = \theta$. Therefore, $\theta$ in
\eqref{rule-big} is alphabetic. Moreover, by direct inspection of (ii)
of the definition of $\stackrel{s}{\Rightarrow}_{\cal N}$ we can see
that $v$ in \eqref{derivation-2} consists of constant functions over
$W$. 
\end{proof}

Now we are able to prove our main composition result, which is in fact the generalization of \cite[Thm. 1]{bak79}.

\begin{theo} \label{comp-lemma} Let $\M$ and $\N$ be s$k$-tt for which the following two conditions hold:
\begin{enumerate}
\item[(a)]  $\M$ is deterministic or $\N$ is linear, and
\item[(b)] $\M$ is total or $\N$ is  nondeleting.
\end{enumerate}
Then the s$k$-tt $\M ; \N$ induces $\M\circ \N$.
\end{theo}
\begin{proof} We prove that, for every $\xi \in T_U\ui k$, $p\in P$, $q\in Q$, and $\zeta \in T_W\ui k$, we have
\begin{equation}\label{one-derivation}
\langle p,q\rangle (\xi) \Rightarrow^*_{\M;\N} \zeta
\end{equation}
if and only if 
\begin{equation}\label{two-derivations}
\text{there exists an } \eta \in T_V\ui k \text{ such that }
q(\xi) \Rightarrow^*_{\M} \eta \text{ and } p(\eta) \Rightarrow^*_{\N} \zeta.
\end{equation}
The proof can be performed by induction on $\xi$. The proof of the
implication (\ref{two-derivations}) $\Rightarrow$
(\ref{one-derivation}) is straightforward, hence we leave it. We note
that (as in \cite[Thm.1]{bak79}) we need neither condition (a) nor (b) for the proof of this direction.

To prove  that (\ref{one-derivation}) $\Rightarrow$ (\ref{two-derivations}),
let us assume that (\ref{one-derivation}) holds and that $\xi=a(\xi_1,\ldots,\xi_l)$ for some $a\in U$, $0\le l \le k$, and $\xi_1,\ldots,\xi_l\in T_U\ui k$.

Let us assume that we applied the rule (\ref{rule-big}) in the first step of (\ref{one-derivation}).
Then  (\ref{one-derivation}) can be written as
\begin{align*}
\langle p,q\rangle (a(\xi_1,\ldots,\xi_l)) & \Rightarrow_{\M;\N} v(a)[\langle p_1, q_{j_1}\rangle (\xi_{i_{j_1}}),\ldots,\langle p_n, q_{j_n}\rangle (\xi_{i_{j_n}})] \\
    &   \Rightarrow^*_{\M;\N} v(a)[\zeta_1,\ldots,\zeta_n]
\end{align*}
for some $\zeta_1,\ldots,\zeta_n\in T_W\ui k$, where $a\in \seml\theta\semr$ and $\zeta=v(a)[\zeta_1,\ldots,\zeta_n]$. Hence,
\[
\langle p_1, q_{j_1}\rangle (\xi_{i_{j_1}}) \Rightarrow^*_{\M;\N}
\zeta_1 \; ,\ldots, \;
\langle p_n, q_{j_n}\rangle (\xi_{i_{j_n}}) \Rightarrow^*_{\M;\N} \zeta_n,
\]
and thus, by the induction hypothesis, there are $\eta_1,\ldots,\eta_n \in T_V\ui k$ such that
\begin{equation}\label{ind-hypothesis}
q_{j_1}(\xi_{i_{j_1}}) \Rightarrow^*_{\M} \eta_1 \text{ and } p_1(\eta_1) \Rightarrow^*_{\N}\zeta_1,\ldots,
q_{j_n}(\xi_{i_{j_n}}) \Rightarrow^*_{\M} \eta_n \text{ and } p_n(\eta_n) \Rightarrow^*_{\N}\zeta_n.
\end{equation}

Since rule (\ref{rule-big}) is in $R$, there is a rule of then form (\ref{rule-1}) in $R_1$ such that the derivation (\ref{derivation-2}) holds. Hence, by Lemma \ref{lift-lemma}, $a\in \seml\varphi \semr$ and
\begin{equation}\label{from-derivation-2}
p(u(a)[q_1(x_{i_1}),\ldots,q_m(x_{i_m})]) \,
(\stackrel{s}{\Rightarrow_{\N}})^* \, v(a)[\langle p_1, q_{j_1}\rangle (x_{i_{j_1}}),\ldots,\langle p_n, q_{j_n}\rangle (x_{i_{j_n}})].
\end{equation}
Now we define the tree $\eta$. For this, let $1\le \lambda \le m$. If $\lambda = j_\alpha$ for some
$1\le \alpha \le n$, then we define $\overline{\eta}_\lambda = \eta_\alpha$. This $\overline{\eta}_\lambda$ is well-defined, which can be seen as follows. Assume that $j_\alpha=j_\beta$
for some $1\le \beta\neq \alpha \le n$.  Then $\N$ is not linear, and thus by condition (a)  $\M$ is deterministic, which implies $\eta_\alpha=\eta_\beta$.
Note that by (\ref{ind-hypothesis})
\[q_\lambda(\xi_{i_\lambda})= q_{j_\alpha}(\xi_{i_{j_\alpha}})\Rightarrow_{\M}^* \eta_\alpha=\overline{\eta}_\lambda.\]
If there is no $\alpha$ with $\lambda = j_\alpha$, then $\N$ is deleting and thus by condition (b) $\M$ is total.
Hence, there is a tree $\overline{\eta}_\lambda\in T_V\ui k$ such that $q_\lambda(\xi_{i_\lambda})
\Rightarrow_{\M}^* \overline{\eta}_\lambda$. 

Let $\eta=u(a)[\overline{\eta}_1,\ldots,\overline{\eta}_m]$. Since the rule (\ref{rule-1}) is in $R_1$ and $a\in \seml\varphi \semr$, we have
\[
q(a(\xi_1,\ldots,\xi_l)) \Rightarrow_{\M} u(a)[q_1(\xi_{i_1}),\ldots,q_m(\xi_{i_m})] \Rightarrow^*_{\M}u(a)[\overline{\eta}_1,\ldots,\overline{\eta}_m].
\]
Moreover by an obvious modification of (\ref{from-derivation-2}) and by (\ref{ind-hypothesis})
\begin{align*}
p(u(a)[\overline{\eta}_1,\ldots,\overline{\eta}_m])& \Rightarrow^*_{\N} v(a)[p_1(\overline{\eta}_{j_1}),\dots,p_1(\overline{\eta}_{j_n})] = \\
v(a)[p_1(\eta_1),\dots,p_1(\eta_n)] & \Rightarrow^*_{\N}v(a)[\zeta_1,\ldots,\zeta_n].
\end{align*}
\end{proof}

Due to Observation \ref{ob:alph-stt} this theorem generalizes \cite[Thm.1]{bak79}.

As an application of the above theorem, we can show that both the class of tree transformations computed by total and deterministic s$k$-tt and the one computed by linear and nondeleting s$k$-tt are closed under composition.

\begin{cor} \rm
\begin{tabular}[t]{lrcl}
(a) & $\tdSTT\ui k\circ \tdSTT\ui k$ & = & $\tdSTT\ui k$ \\
(b) & $\lnSTT\ui k\circ \lnSTT\ui k$ & = & $\lnSTT\ui k$. 
\end{tabular}
\end{cor}
\begin{proof} We prove only (a) because the proof of (b) is similar. The inclusion from left to right can be seen as follows.
Let $\M$ and $\N$ be total and deterministic s$k$-tt. The s$k$-tt $\M;\N$ is also 
total and deterministic and, by Theorem \ref{comp-lemma}, for the computed tree transformations
$\M;\N=\M\circ \N$ holds. The other inclusion follows from the facts
that (i) any tree transformation $\tau \subseteq T_U \ui k\times T_V
\ui k$ can be decomposed as $\tau \circ \iota_{T_V \ui k}$ and (ii)~$\iota_{T_V \ui k}$ can be computed by a total and deterministic s$k$-tt.
\end{proof}

\section{Forward and backward application of stt}

In this section we consider forward and backward application of stt to s-recognizable tree languages.
In particular, we consider the domain and the range of tree transformations computed by stt. 
Finally, we apply these results to the problem of (inverse) type checking. 

\subsection{Application of stt}

We begin with the following result.

\begin{theo} \label{dom-theo}$\dom(\STT\ui k) = \REC\ui k$.
\end{theo}
\begin{proof} First we prove the inclusion from left to right. For
  this, let $\M = (Q,U,\Phi,V,q_0,R)$ be an s$k$-tt. We construct the
  s$k$-rtg $\G = (\PS(Q),U,\Phi,\{q_0\},R')$ such that $\dom(\M) = L(\G)$, where the set
$R'$ of rules is defined as follows.
 
For every $0 \le l \le k$ and  $P \subseteq Q$ with $P = \{p_1,\ldots,p_m\}$ for some
$m \ge 1$, and rules
\begin{equation}
 p_1(\varphi_1(x_1,\ldots,x_l)) \rightarrow u_1,
  \ldots, p_m(\varphi_m(x_1,\ldots,x_l)) \rightarrow u_m \label{rules} 
\end{equation}
in $R$, let $R'$ contain the rule
\[
P\rightarrow (\varphi_1 \wedge \ldots \wedge \varphi_m)(P_1,\ldots,P_l)
\]
where $P_i = \{q \in Q \mid q(x_i) \text{ occurs in } u_j \text{ for
  some $1 \le j \le m$}\}$. Thus, in particular, for every $0\le l\le k$, the rule
\[
\emptyset \rightarrow
\top(\emptyset,\ldots,\emptyset)
\]
with $l$ occurrences of $\emptyset$ in its right-hand side is in $R'$. Hence $\emptyset \Rightarrow_\G^* \xi$ for every $\xi\in T_U^{(k)}$.

We claim that for every $P \subseteq Q$ and $\xi \in T_U$ we have:
\begin{equation}
P \Rightarrow_\G^* \xi\text{  iff $\Big($ for every $p \in P$ there is a $\zeta \in T_V$ such that } p(\xi)
\Rightarrow_\M^* \zeta\Big). \label{equ:dom-G}
\end{equation}
The statement is clear for $P=\emptyset$, therefore we assume that $P = \{p_1,\ldots,p_m\}$ for some
$m \ge 1$. We prove (\ref{equ:dom-G}) by induction on $\xi$.

Let $\xi = a(\xi_1,\ldots,\xi_l) \in T_U^{(k)}$. Then

\

\begin{tabular}{cl}
& $P \Rightarrow_\G a(P_1,\ldots,P_l)\Rightarrow_\G^*
a(\xi_1,\ldots,\xi_l)$ \\[3mm]

iff & there is a rule  $P\rightarrow (\varphi_1 \wedge \ldots \wedge \varphi_m)(P_1,\ldots,P_l)$ in $R'$ with $a\in \big(\bigcap_{j=1}^m\seml \varphi_j\semr\big)$\\[1mm]

& and for every $1\le i\le l$, $P_i \Rightarrow_\G^* \xi_i$ \\[3mm]

iff & there are rules (\ref{rules}) in $R$ with $a\in \big(\bigcap_{j=1}^m\seml \varphi_j\semr\big)$ and\\[1mm]

&  for every $1\le i\le l$ and $q\in P_i$ there is a tree $\zeta_{i,q} \in T_V$ s.t. $q(\xi_i)
\Rightarrow_\M^* \zeta_{i,q}$\\[3mm]

iff & for every $1\le j\le m$  there is a rule $p_j(\varphi_j(x_1,\ldots,x_l)) \rightarrow u_j$ s.t. $a\in \seml \varphi_j\semr$ \\[1mm]

& and  for every occurrence of $q(x_i)$ in $u_j$ $\exists$ a tree $\zeta_{i,q} \in T_V$ s.t. $q(\xi_i)
\Rightarrow_\M^* \zeta_{i,q}$\\[3mm]

iff & for every $1\le j\le m$ there is  a tree $\zeta_{j} \in T_V$ such that $p_j(a(\xi_1,\ldots,\xi_l))
\Rightarrow_\M^* \zeta_{j}$.
\end{tabular}

\

\noindent Statement \eqref{equ:dom-G} with $P=\{q_0\}$ implies  $L(\G) = \dom(\M)$. Hence, by Theorem \ref{th:rec=reg} we obtain that $\dom(\M)$ is s$k$-recognizable.
The other inclusion follows from Lemma \ref{ln-lemma}.
\end{proof}

\begin{ex}\rm We illustrate the construction of the s$k$-rtg $\G$ in the proof
  of Theorem \ref{dom-theo} by an example. 

Let the s2-tt $\M$ contain the rules
\[
\begin{array}{rclrcl}
q_0(\varphi(x_1,x_2)) & \rightarrow & f(p(x_1),q(x_1)) \hspace{10mm} &
\bar{p}(\theta_1) & \rightarrow & h_1\\
p(\psi(x_1)) & \rightarrow & g(\bar{p}(x_1),p'(x_1)) \hspace{10mm} &
p'(\theta_2) & \rightarrow & h_2\\
q(\psi'(x_1)) & \rightarrow & \hat{p}(x_1) \hspace{10mm} &
\hat{p}(\theta_3) & \rightarrow & h_3\\
\end{array}
\]

Then s2-rtg $\G$ contains (among others) the following rules:

\[
\begin{array}{rclrcl}
\{q_0\} &\rightarrow& \varphi\big(\{p,q\}, \emptyset
\big) \hspace{10mm}&
\emptyset&\rightarrow& \top\\
\{p,q\} &\rightarrow& (\psi \wedge
\psi')\big(\{\bar{p},p',\hat{p}\}\big) \hspace{10mm}&
\emptyset&\rightarrow& \top\big(\emptyset \big) \\
\{\bar{p},p',\hat{p}\} &\rightarrow& (\theta_1 \wedge \theta_2
\wedge \theta_3) \hspace{10mm}&
\emptyset&\rightarrow& \top\big(\emptyset,\emptyset\big).
\end{array} 
\]
\end{ex}

Now we can prove that backward application of stt preserve recognizability of tree languages.

\begin{theo} \label{backward-theo}$(\STT\ui k)^{-1}(\REC\ui k) = \REC\ui k$.
\end{theo}
\begin{proof} First we prove the inclusion from left to right. For this, let $\M = (Q,U,\Phi,V,q_0,R)$ be an s$k$-tt, and $L\subseteq T_V\ui k$ an s$k$-recognizable tree language. It is an elementary fact that $\M^{-1}(L) = \dom(\M\circ\iota_L)$. By Lemma \ref{ln-lemma}, there is a linear and nondeleting s$k$-tt $\N$
with $\N = \iota_L$. Moreover, by Theorem \ref{comp-lemma}, the s$k$-tt $\M;\N$ induces $\M\circ \N$. Hence $\M^{-1}(L) = \dom(\M;\N)$, which is s$k$-recognizable by Theorem \ref{dom-theo}.

The other inclusion follows from Lemma \ref{ln-lemma}.
\end{proof} 

It is well-known from the theory of classical tree automata and tree transducers that the forward application of linear top-down tree transformations preserve recognizability of tree languages (see e.g. \cite{tha69} or \cite[Ch. IV, Cor. 6.6]{gecste84}). In particular, the range of every linear top-down tree transformation is a recognizable tree language. We can show easily that a linear s$k$-tt does not have the analogous property.

\begin{lm} \label{range-lin-lemma}\rm There is a linear s1-tt $\M$ such that $\range(\M)$ is not 1-recognizable.
\end{lm}
\begin{proof} Let us assume that $U$ is infinite and define the s1-tt $\M = (\{q\},U,\{\top\},U,q,R)$, where $R$ consists of the only rule
$$q(\top(\,))\to \iota_U(\iota_U).$$
It is clear that $\M$ induces the 1-tree transformation $\{(a,a(a))\mid a\in U\}$. Thus $\range(\M)= \{a(a)\mid a\in U\}$, which is not 1-recognizable
by the remark after Lemma \ref{non-rec-lemma}.
\end{proof}

The non-recognizability of $\range(\M)$ above is due to the fact that
$\M$ is able to ``duplicate" a node of the input tree by having two
occurrences of an appropriate function symbol on the right-hand side
of its rule. We would like to identify a restricted version of an stt
which does not have this capability in the hope of that such an stt
preserves recognizability. Therefore we define simple stt as
follows. An s$k$-tt $\M = (Q,U,\Phi,V,q_0,R)$ is {\em simple} if
$\rhs(\rho)$ contains exactly one function symbol for every rule $\rho
\in R$. We denote the class of tree transformations computed by simple
and linear stt by $\slSTT$. Then we can prove the desired result using
the following notation. If $\varphi \in \Pred(U)$ and $f: U
\rightarrow V$ is a mapping, then $f(\varphi)$ denotes the
predicate defined by $\seml f(\varphi) \semr = f(\seml \varphi \semr)$.

\begin{theo} \label{slin-theo}$\slSTT\ui k(\REC\ui k) = \REC\ui k$.
\end{theo}
\begin{proof} First we prove the inclusion from left to right. Let $\M = (Q,U,\Phi,V,q_0,R)$ be a simple and linear s$k$-tt and 
$L$ be an s-recognizable tree language such that $L=L(\G)$ for some reduced s$k$-rtg $\G=(P,U,\Psi,p_0,R_\G)$ which is
in normal form (cf. Theorem \ref{th:rec=reg} and Lemma \ref{lm:srtg-normal-form}).

We construct the s$k$-rtg $\G'=(Q\times P,V,\Psi',\langle q_0,p_0\rangle,R')$, where
\begin{itemize}
\item $\Psi'=\{ f(\varphi\wedge \psi) \mid \varphi \text{ and } f \text{ occur
 in a rule of } R, \text{ and } \psi \text{  in a rule of } R_\G \}$, and
\item $R'$ is the smallest set of rules satisfying that if $p \to \psi(p_1,\ldots,p_l)$ is in $R_\G$ and $q(\varphi(x_1,\ldots,x_l))\to f(q_1(x_{i_1}),\ldots,q_m(x_{i_m}))$ is in $R$, then 
the rule 
\begin{equation}
\langle q,p\rangle \to f(\varphi\wedge \psi)(\langle q_1,p_{i_1}\rangle,\ldots,\langle q_m,p_{i_m}\rangle)\label{rule}
\end{equation}
 is in $R'$.
\end{itemize} 
We show that $L(\G')=\M(L)$. For this it suffices
to prove the following statement. For every $q\in Q$, $p\in P$, and $\zeta \in T_V$ we have
\[ \langle q,p\rangle \Rightarrow_{\G'}^* \zeta \iff \exists (\xi \in L(\G,p)) \text{ such that } q(\xi) \Rightarrow_\M^* \zeta.\]
We prove only the direction $\Rightarrow$ by induction on the number $n$ of steps of the corresponding derivation and we show only the induction step $n$ to $n+1$. The other direction can be proved in a similar way.

\underline{Direction $\Rightarrow$, step $n\to n+1$:} We assume that in the first step of the derivation we applied the rule (\ref{rule}) obtained from the rules
$p \to \psi(p_1,\ldots,p_l)$ in $R_\G$ and $q(\varphi(x_1,\ldots,x_l))\to f(q_1(x_{i_1}),\ldots,q_m(x_{i_m}))$  in $R$. (Note that $\{i_1,\ldots,i_m\} \subseteq \{1,\ldots,l\}$.) Then we have
$$\langle q,p\rangle \Rightarrow_{\G'} b(\langle q_1,p_{i_1}\rangle,\ldots,\langle q_m,p_{i_m}\rangle)  \Rightarrow_{\G'}^n b(\zeta_1,\ldots,\zeta_m)$$
for some $b\in \seml f(\varphi\wedge \psi)\semr$ and $\zeta_1,\ldots,\zeta_m \in T_V\ui k$. By the I.H., there are trees  $\xi_{i_j} \in  L(\G,p_{i_j}))$
such that $q_{i_j}(\xi_{i_j}) \Rightarrow_\M^* \zeta_j$ for every $1\le j\le m$. Moreover, there is a $a\in (\seml \varphi\semr \cap \seml \psi\semr)$ such that $b=f(a)$. Now define the tree $\xi=a(\overline{\xi_1},\ldots,\overline{\xi_l})\in T_U\ui k$, where $\overline{\xi_j}=\xi_{i_l}$ if $j=i_l$ for some $1\le l\le m$; and let $\overline{\xi_j}$ be an arbitrary tree in $L(\G,p_j)$ otherwise (note that $\G$ is reduced). Then 
$$p\Rightarrow_G a(p_1,\ldots,p_l) \Rightarrow_G^* a(\overline{\xi_1},\ldots,\overline{\xi_l}),$$
hence $\xi \in  L(\G,p)$.  Moreover
$$q(a(\overline{\xi_1},\ldots,\overline{\xi_l}))\Rightarrow_\M b(\langle q_1,p_{i_1}\rangle (\xi_{i_1}),\ldots,\langle q_m,p_{i_m}\rangle (\xi_{i_m}))\Rightarrow_\M^* b(\zeta_1,\ldots,\zeta_m).$$

The inclusion from right to left follows from Lemma \ref{ln-lemma} and the fact that $\A_{=}$ is a simple and linear stt.
\end{proof}

\begin{cor} \rm \label{range-lemma}$\range(\slSTT\ui k)=\REC\ui k$.
\end{cor}
\begin{proof}  Let $\M = (Q,U,\Phi,V,q_0,R)$ be a simple and linear s$k$-tt. Obviously, $\range(\M)=\M(T_U\ui k)$. Moreover, by Observation \ref{ob:TU-k-rec}, $T_U\ui k$ is s$k$-recognizable. Hence the statement follows from Theorem \ref{slin-theo}.
\end{proof}

\subsection{Type checking with stt}

Intuitively, type checking means to verify whether or not all documents in a view have a certain type. According to \cite{engman03}, a typical scenario of type checking is that  $\tau$ translates XML documents into HTML documents. Thus, for a set $L$ of XML documents $\tau(L)$ is an HTML-view of the documents in $L$. In practice, we are interested in particular XML documents, which turn to be a recognizable tree language of unranked trees over some alphabet. 
Also, certain desired properties of the so-obtained HTML documents can be described in terms of recognizability of tree languages. Thus, the type checking problem of $\tau$ in fact means to check whether $\tau(L)\subseteq L'$ for recognizable   tree languages $L$ and $L'$. The inverse type checking problem can be described in a similar way. The type checking and the inverse type checking problem for different kinds of transducers was considered in several works, see among others \cite{milsucvia03,alomilnevsucvia03,engman03,manberpersei05}.
For stt we obtain the following results.

\begin{theo}\label{thm:type-checking} 

\

\begin{enumerate}
\item[(a)] The inverse type checking problem for stt is decidable.
\item[(b)] The type checking problem for simple and linear stt is decidable. 
\end{enumerate}
\end{theo}
\begin{proof}Both statements follow from the fact that the inclusion problem of s-recognizable tree languages is decidable. This latter fact can be seen as follows. By \cite[Thm. 3]{veabjo11a}, s-recognizable tree languages are effectively closed under Boolean operations, for closure under complement, see our correction at the end of Section \ref{sect:sta-def}. Moreover, by \cite[Thm. 4]{veabjo11a}, the emptiness problem is decidable for s-recognizable tree languages provided that the emptiness problem in the underlying label structure is decidable. Since, by our definition, the label structure underlying an s$k$-ta has a decidable emptiness problem, we obtain that the inclusion problem of s-recognizable tree languages is decidable.

Then the proof of (a) is as follows. Let $\M:T_U\ui k\to T_V\ui k$ be an s$k$-tt and $L'\subseteq T_U\ui k$ and $L\subseteq T_V\ui k$ s-recognizable tree languages.  By Theorem \ref{backward-theo},  the tree language  $\M^{-1}(L)$ is effectively s$k$-recognizable, thus we can decide if
$\M^{-1}(L)\subseteq L'$ holds or not. Statement (b) can be proved in a similar way, using Theorem \ref{slin-theo}.
\end{proof}

\section{Conclusion and an open problem}

In this paper we have further elaborated the theory of sta and stt. Our main contributions are: the characterization of s-recognizable tree languages in terms of relabelings of recognizable tree languages, the introduction of symbolic regular tree grammars and the proof of their equivalence to sta, the comparison of sta and variable tree automata, the composition of stt, and the forward and backward application of stt to s-recognizable tree languages.

Finally, we  mention an open problem.
In the definition of simple s$k$-tt we required that the right-hand side of each rule contains exactly one function symbol. We conjecture that, for the closure result in Theorem \ref{slin-theo}, it is sufficient to require that right-hand sides of rules contain \underline{at most one} function

\bibliographystyle{alpha}
\newcommand{\etalchar}[1]{$^{#1}$}

\end{document}